\documentclass[twoside,11pt]{article}

\usepackage{blindtext}

\usepackage{jmlr2e}

\usepackage{mypackages}
\usepackage{textcomp}
\usepackage[dvipsnames]{xcolor}

\usepackage{lastpage}
\ShortHeadings{The Symmetric alpha-Stable Privacy Mechanism}{Zawacki and Abed}
\firstpageno{1}

\begin{document}

\title{Heavy-Tailed Privacy: The Symmetric alpha-Stable Privacy Mechanism}

\author{\name Christopher C. Zawacki \email czawacki@umd.edu \\
       \addr Dept. of Electrical and Computer Engineering\\
       University of Maryland\\
       College Park, MD 20742 USA
       \AND
       \name Eyad H. Abed \email abed@umd.edu \\
       \addr Dept. of Electrical and Computer Engineering\\
       University of Maryland\\
       College Park, MD 20742 USA}

\editor{Christopher C. Zawacki and Eyad H. Abed}

\maketitle

\begin{abstract}%   <- trailing '%' for backward compatibility of .sty file
With the rapid growth of digital platforms, there is increasing apprehension about how personal data is collected, stored, and used by various entities. These concerns arise from the increasing frequency of data breaches, cyber-attacks, and misuse of personal information for targeted advertising and surveillance. To address these matters, Differential Privacy (DP) has emerged as a prominent tool for quantifying a digital system's level of protection. The Gaussian mechanism is commonly used because the Gaussian density is closed under convolution, and is a common method utilized when aggregating datasets. However, the Gaussian mechanism only satisfies an approximate form of Differential Privacy. In this work, we present and analyze of the Symmetric alpha-Stable (SaS) mechanism. We prove that the mechanism achieves pure differential privacy while remaining closed under convolution. Additionally, we study the nuanced relationship between the level of privacy achieved and the parameters of the density. Lastly, we compare the expected error introduced to dataset queries by the Gaussian and SaS mechanisms.  From our analysis, we believe the SaS Mechanism is an appealing choice for privacy-focused applications. 
\end{abstract}


\section{INTRODUCTION}
\label{sec:introduction}

Privacy is fundamental to individual autonomy and personal safety. It protects individuals from harassment and discrimination, fosters trust in institutions, and encourages free speech and innovation. As the world becomes increasingly digital, we have seen in \cite{itrc22} a corresponding increase in data breaches targeting the growing number of individual databases that hold client information. In recent years, the public and private sectors have begun to act. Political leaders are taking action to ensure the privacy of their citizens, for example the Internet Freedom Act \cite{transparency11} and the General Data Protection Regulation \cite{GDPR18}, and consumers are putting pressure on companies to adopt settings and methods that focus on the privacy of their customers, as discussed in \cite{koetsier21} and \cite{minto21}.

Introduced by \cite{dwork06, dwork06b}, one common approach to protecting client data is known as Differential Privacy (DP). Differentially private systems inject carefully constructed noise into a dataset to obfuscate the specific participants, while retaining general trends in machine learning datasets. The Differential Privacy framework has been applied in various domains, from large-scale data analysis to machine learning; see \cite{abadi16}. More recently, Differential Privacy has received renewed attention within the field of federated learning (FL), a privacy focused branch of machine learning introduced by \cite{mcmahan16}. The objective of differentially private FL methods are to enhance privacy preservation while collaboratively training machine learning models across multiple decentralized devices or servers \cite{wei20}. In \cite{li19b}, the authors use differentially private federated learning methods to train a machine learning model that segments images of brain tumors. 

The work here extends the initial results presented in \cite{zawacki24} with additional insight into the how the expected error increases as the level of noise increases. Related to the theme of this work, other groups have begun to examine the benefits of using heavy-tailed distributions within the differential privacy framework. \cite{ito21} use heavy-tailed distributions to mask contributions by outliers in the of filter/controller design for control systems. In \cite{asi24}, the authors determine optimal rates of convergence (up to a logarithm) in the context of private convex optimization with heavy tailed distributions. \cite{csimcsekli24} show that under broad conditions, the use of heavy-tailed distributions in differentially private stochastic gradient descent (SGD) eliminates the need for a projection step, decreasing the computational complexity. Our results differ in the level of privacy guaranteed by the privacy mechanism and we additionally provide a deeper analysis on the relationship between the parameters of the density and their effect on the level of protection. 

The contributions of this work are threefold. First we introduce our new privacy mechanism which is based on the use of stable densities and prove that this mechanism is $\varepsilon$-differentially private. Second, we show that the level of privacy scales inversely with the level of injected noise; aligning its behavior with existing privacy mechanisms. Lastly, we compare the expected distortion of our privacy mechanism against other commonly utilized privacy mechanisms.

The rest of the paper is organized as follows. Section \ref{sec:background} summarizes the basics of Differential Privacy. Section \ref{sec:sas} introduces the definition of the Symmetric alpha-Stable Mechanism. Section \ref{sec:priv} proves the privacy guarantee of the new mechanism. Section \ref{sec:scale} studies how the privacy scales with the level of noise. Section \ref{sec:error} provides a measure of error for the mechanism introduces. Section \ref{sec:conclusion} summarizes the results and provides comments on active related research efforts.

\section{BACKGROUND}
\label{sec:background}

In this section we outline the background material required to derive our results. 

\subsection{Differential Privacy}
Differential Privacy operates on a collaboratively constructed dataset, which we denote by $\mathcal{D}$. Conceptually, we can think of such a dataset as a table of records, where each row represents a set of client data. Denote by $f$ a function that operates on this dataset and produces a vector of $m$ numerical values. For instance:

\begin{itemize}
\item How many clients have blue eyes?
\item What is the average income of all clients?
\item What are the optimized parameters of a given machine learning model across all clients?
\end{itemize}

Differential Privacy primarily prevents \textit{passive} adversaries from obtaining information about a target client. A passive adversary is one that observes communications or interactions without actively tampering with them. For example, a passive adversary may eavesdrop on communication between a client and the server or combine publicly available datasets in a linking attack to re-identify anonymized client data. This is in contrast with an active adversary which seeks to directly disrupt model training.

By a slight abuse of notation, we use the symbol $f$ for the query, despite the composition of the dataset,
\begin{definition}
  \label{def:query}
  (Query)
  A function $f$ is termed a query if it takes a dataset $\Dc$ as input and outputs a vector in $\mathbb{R}^m$:
  \begin{equation}
    f: \Dc \to \mathbb{R}^m.
  \end{equation}
  \vspace{-6mm}
  \begin{equation*}\tag*{\textrm{$\blacktriangleleft$}}\end{equation*}
\end{definition}

The types of queries commonly employed in Differential Privacy methods are those exhibiting finite $\ell_p$-sensitivity \cite{dwork06, roth14}:
\begin{definition}
  \label{def:lp}
  ($\ell_p$-Sensitivity of Query) The $\ell_p$-sensitivity of a query $f$, denoted $\Delta_pf$, is defined to be a maximum of a $p-$norm over the domain of $f$, $dom(f)$:
  \begin{equation}
    \Delta_pf := \max_{\Dc_1\simeq \Dc_2} ||f(\Dc_1) - f(\Dc_2)||_p,
  \end{equation}
  for all $\Dc_1, \Dc_2 \in dom(f)$.
  \vspace{-2mm}
  \begin{equation*}\tag*{\textrm{$\blacktriangleleft$}}\end{equation*}
\end{definition}

From Definitions \ref{def:query} and \ref{def:lp} above, it is clear that when the sensitivity of $f$ is bounded, the range of the query is also bounded. While focusing on finite queries here, ongoing research aims to extend differentially private methods to handle queries with unbounded ranges, see \cite{durfee24}.

Now, we recall the definition of a privacy mechanism, which introduces stochastic noise to the result of a query.
\begin{definition}
  (Privacy mechanism) A privacy mechanism for the query $f$, denoted $\Mf$, is defined to be a randomized algorithm that returns the result of the query perturbed by a vector of i.i.d. noise sampled from pre-selected densities $Y_i$,
  \begin{equation}
    \Mf(\Dc) = f(\Dc) + (Y_1, Y_2, \dots, Y_m)^T.
  \end{equation}
  \vspace{-6mm}
  \begin{equation*}\tag*{\textrm{$\blacktriangleleft$}}\end{equation*}
\end{definition}

To simplify notation, we denote the resulting vector, $\Mf(\Dc)$, as $\x\in\mathbb{R}^m$. Note that the noise variables, $Y_i$, induce a density, which we occasionally denote $p = p(\x)$ for $\Mf$, on a given dataset $\Dc$. Although not strictly necessary, we assume that the injected density is symmetric about the origin, simplifying the analysis. 

The privacy mechanism aims to hinder an adversary from conclusively ascertaining the presence of a specific client within the dataset.
\begin{definition}
  (Neighboring Datasets) Two datasets, denoted $\Dc_1$ and $\Dc_2$, are known as neighboring datasets if they differ in the presence or absence of exactly one client record. We denote this relation as $\Dc_1\simeq \Dc_2$.
  \vspace{-6mm}
  \begin{equation*}\tag*{\textrm{$\blacktriangleleft$}}\end{equation*}
\end{definition}

This concept is visualized in Figure \ref{fig:nghb-datasets}, which depicts two datasets, one that contains the red client, and one that does not. Let $\Dc_1$ and $\Dc_2$ represent these two scenarios respectively. To proceed, let us assume the red client has allowed their data to be included in the set and that $\Dc_1$ is the \textit{true} dataset. Denote a realization of a mechanism as $x \sim \Mf(\Dc_a)$. Informally, the mechanism $\Mf$ is said to be differentially private if the inclusion or exclusion of a single individual in the dataset, illustrated in red in the figure, results in \textit{essentially} the same distribution over the realized outputs $x$,
\begin{equation}
  \Pr[\Mf(\Dc_1) = x] \approx \Pr[\Mf(\Dc_2) = x].
\end{equation}
\begin{figure}[ht!]
  \centering
  \includegraphics[width=0.6\linewidth]{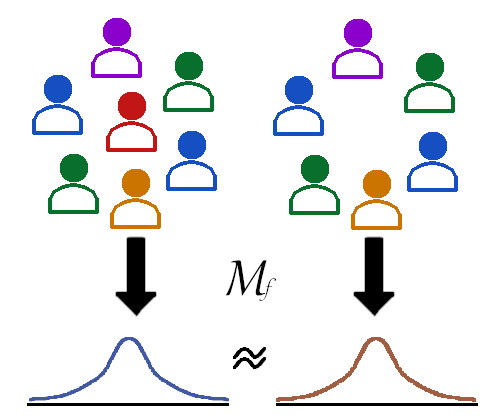}
  \caption{\label{fig:nghb-datasets} In order to protect client identity, Differential Privacy injects noise into the output of a query $f$ on a dataset. This induces a probability density over possible outcomes. A mechanism, $\mathcal{M}_f$, is considered private, if the resulting distributions are \textit{essentially} the same regardless of the inclusion or exclusion of a single client, shown here in red. Differential Privacy quantifies how much information an adversary is able to gain about the red client.}
\end{figure}
Differential Privacy then quantifies what \textit{essentially} means mathematically:
\begin{definition}
  (Pure Differential Privacy)
  Let $\Dc_1$ and $\Dc_2$ be any neighboring datasets. Given a query $f$ that operates on $\Dc_1$ and $\Dc_2$, a privacy mechanism $\Mf$ is said to be $\varepsilon$-Differentially Private ($\varepsilon$-DP or pure-DP) if it satisfies
  \begin{equation}
    \label{eqn:pureDP}
    \Pr[\Mf(\Dc_1) \in \mathcal{X}] \leq e^\varepsilon \Pr[\Mf(\Dc_2) \in \mathcal{X}]
  \end{equation}
  for some $\varepsilon > 0$ and any subset of outputs $\mathcal{X} \subseteq \mathcal{R}(\Mf(\Dc_1))$. The mechanism is defined to have no privacy ($\epsilon = \infty$) if, upon its application to each dataset, the supports of the resulting densities are not equal, viz. $\mathcal{R}(\Mf(\Dc_1)) \neq \mathcal{R}(\Mf(\Dc_2))$.
  \vspace{-2mm}
  \begin{equation*}\tag*{\textrm{$\blacktriangleleft$}}\end{equation*}
\end{definition}

The parameter $\varepsilon$ is also referred to as the privacy budget. Smaller values of $\varepsilon$ are associated with stronger privacy. We remark that when $\varepsilon=0$, the definition yields perfect privacy. However, in that case, adding more client data results in no new information.

Note that Eq. \ref{eqn:pureDP} holds for each element when the density of the distributions is considered in \cite{roth14}:
\begin{theorem}
  \label{thm:atoms} (Privacy as Densities)
  Let $\Dc_1$ and $\Dc_2$ be neighboring datasets and $f$ be a query that operates on them. Denote by $p_1$ and $p_2$ the densities of the privacy mechanism $\Mf$ when applied to $\Dc_1$ and $\Dc_2$ respectively. Then, a privacy mechanism $\Mf$ is $\varepsilon$-Differentially Private if
  \begin{equation}
    \label{eqn:purepdf}
    p_1(x) \leq e^\varepsilon p_2(x), \quad \forall x \in \mathcal{R}(\Mf(\Dc_1))
  \end{equation}
  for all $\Dc_1 \simeq \Dc_2$. 
  \vspace{-4mm}
  \begin{equation*}\tag*{\textrm{$\blacktriangleleft$}}\end{equation*}
\end{theorem}
Next, we give a brief proof for this known fact.
\begin{proof}
  Begin by writing condition (\ref{eqn:pureDP}) in terms of the generated densities,
  \begin{equation}
    \label{eqn:thm1-a}
    \int_\mathcal{X} p_1(x) dx \leq \int_\mathcal{X} e^\epsilon p_2(x)dx.
  \end{equation}
  Equation (\ref{eqn:thm1-a}) can be rewritten as
  \begin{equation}
    \label{eqn:thm1-b}
    0 \leq \int_\mathcal{X} e^\epsilon p_2(x) - p_1(x) dx.
  \end{equation}
  Noting that (\ref{eqn:purepdf}) enforces the integrand in (\ref{eqn:thm1-b}) to be non-negative, implying that (\ref{eqn:pureDP}) is satisfied.
\end{proof}

Next, we recall a metric for evaluating the loss of privacy experienced by a participating client under a given privacy mechanism.

\begin{definition}   
  \label{eqn:privloss}
(Privacy Loss) The privacy loss of an outcome $x$ is defined to be the log-ratio of the densities when the mechanism is applied to $\mathcal{D}_1$ and $\mathcal{D}_2$ at $x$ \cite{roth14}:
\begin{equation}
\label{eqn:plf}
 \mathcal{L}_{\mathcal{D}_1 || \mathcal{D}_2}(x) := \ln \frac{p_1(x)}{p_2(x)}.
\end{equation}
By (\ref{eqn:purepdf}), it is evident that $\varepsilon$-Differential Privacy (\ref{eqn:pureDP}) is equivalent to
\begin{equation}
    |\mathcal{L}_{\mathcal{D}_1||\mathcal{D}_2}(x)| \leq \varepsilon, \quad \forall x \in \mathcal{R}(\mathcal{M}_f(\mathcal{D}_1))
\end{equation}
for all neighboring datasets $\mathcal{D}_1$ and $\mathcal{D}_2$. 
\vspace{-6mm}
\begin{equation*}\tag*{\textrm{$\blacktriangleleft$}}\end{equation*}
\end{definition}
For mechanisms that are purely differential private, the privacy budget $\varepsilon$ is the maximum over all observations $x$,
\begin{equation}
\label{eqn:epmax}
    \varepsilon = \max_{x\in\mathbb{R}}\mathcal{L}_{\mathcal{D}_1||{D}_2}(x).
\end{equation}

Due to its beneficial mathematical properties, the Gaussian density is a commonly chosen density for Differential Privacy. However, the Gaussian mechanisms fails to satisfy condition (\ref{eqn:pureDP}). To accommodate this, the condition can be relaxed through the inclusion of an additive constant $\delta > 0$, as in the following definition:
\begin{definition} (Approximate Differential Privacy)
  Let $\Dc_1$ and $\Dc_2$ be any neighboring datasets. Given a query $f$ that operates on $\Dc_1$ and $\Dc_2$, a privacy mechanism $\Mf$ is said to be $(\varepsilon, \delta)$-Differentially Private if it satisfies
  \begin{equation}
    \label{eqn:approxDP}
    \Pr[\Mf(\Dc_1) \in \mathcal{X}] \leq e^\varepsilon \Pr[\Mf(\Dc_2) \in \mathcal{X}] + \delta.
  \end{equation}
  This is known as approximate-Differential Privacy.
  \vspace{-4mm}
  \begin{equation*}\tag*{\textrm{$\blacktriangleleft$}}\end{equation*}
\end{definition}
The accepted error term $\delta$ is the probability that the result of the query provides more information to the adversary than expected from the bound $\varepsilon$.

One common modification relates to who applies the privacy mechanism. Up to this point, we have considered a mechanism in relation to a query over the entire dataset $\Dc$. It is then understood that the mechanism is applied by a \textit{trusted aggregator}, who collects the clients' data prior to obfuscation. However, there does not always exist such a trusted central authority. For example, in a Federated Learning framework, the server is assumed untrustworthy by default. Another situation where clients may with to apply noise locally is if they lack secure communication protocols. In this case, the clients' sensitive information could be leaked to an adversary during the transmission between the clients and server. To this end, a mechanism $\Mf$ is said to be Locally Differentially Private (LDP) if the mechanism can be applied locally by the clients prior to transmission to the server.
\begin{definition} (Local Differential Privacy)
  \label{def:local}
 Let a client apply the privacy mechanism $\mathcal{M}^{loc}_f$ to their local dataset $\Dc$. The mechanism $\mathcal{M}^{loc}_f$ is said to be locally differentially privacy if, for any pair of data points $v_1, v_2 \in \Dc$, it satisfies the following \cite{kasiviswanathan11}:
  \begin{equation}
    \label{eqn:LDP}
    \Pr[\mathcal{M}^{loc}_f(v_1) \in \mathcal{X}] \leq e^\varepsilon \Pr[\mathcal{M}^{loc}_f(v_2) \in \mathcal{X}] + \delta,
  \end{equation}
  for all $\mathcal{X} \in \mathcal{R}(\mathcal{M}^{loc}_f)$.
  \vspace{-4mm}
  \begin{equation*}\tag*{\textrm{$\blacktriangleleft$}}\end{equation*} 
\end{definition}

The mechanism is called $\varepsilon$-LDP if $\delta=0$ and $(\varepsilon, \delta)$-LDP otherwise.

\subsection{Selecting a Level of Privacy}
\cite{wasserman09, geng15} describe a useful connection between Differential Privacy and hypothesis testing. Their analysis considers the problem of client privacy from the perspective of an adversary deciding between two hypothesizes. Denote by $\mathcal{D}_1$ and $\mathcal{D}_2$ two neighboring datasets. Let one of the following hypothesizes hold:
\begin{itemize}
    \item $H_0$ (The null hypothesis): the true dataset is $\mathcal{D}_1$.
    \item $H_1$ (The alternative hypothesis): the true dataset is $\mathcal{D}_2$.
\end{itemize}
The objective of the adversary is to determine, based on the output of a privacy mechanism $\mathcal{M}_f$, which hypothesis is true. Denote by $p$ the probability of a false positive, that is, the adversary chooses $H_1$ when $H_0$ is true. Then, denote by $q$ the probability of a false negative, i.e., $H_0$ is chosen when $H_1$ is true. The authors show that if a mechanism $\mathcal{M}_f$ is $\varepsilon$-differentially private, then the following two statements must hold:
\begin{equation}
\label{eqn:h0}
    p + e^\varepsilon q \geq 1 \textrm{ and }  e^{\varepsilon}p + q \geq 1.
\end{equation}
% and 
% \begin{equation}
% \label{eqn:h1}
   
% \end{equation}
Combining the inequalities in (\ref{eqn:h0}) yields
\begin{equation}
\label{eqn:h2}
    p + q \geq \frac{2}{1 + e^\varepsilon}.
\end{equation}
Consider that when $\varepsilon << 1$, which equates to high privacy, the adversary cannot achieve both low false positive and low false negative rates simultaneously. Often, it is more convenient to specify lower bounds for $p$ and $q$ and to use (\ref{eqn:h2}) to determine $\varepsilon$ than it is to state the privacy budget directly. When a mechanism that satisfies pure Differential Privacy is employed, the maximum information an adversary may learn is strictly bounded. Figure \ref{fig:decision-theory} depicts the upper and lower bounds for a given privacy budget $\varepsilon$ and initial probability $p_1$ on a simple dataset.
\begin{figure}[ht!]
  \centering
  \includegraphics[width=0.6\linewidth]{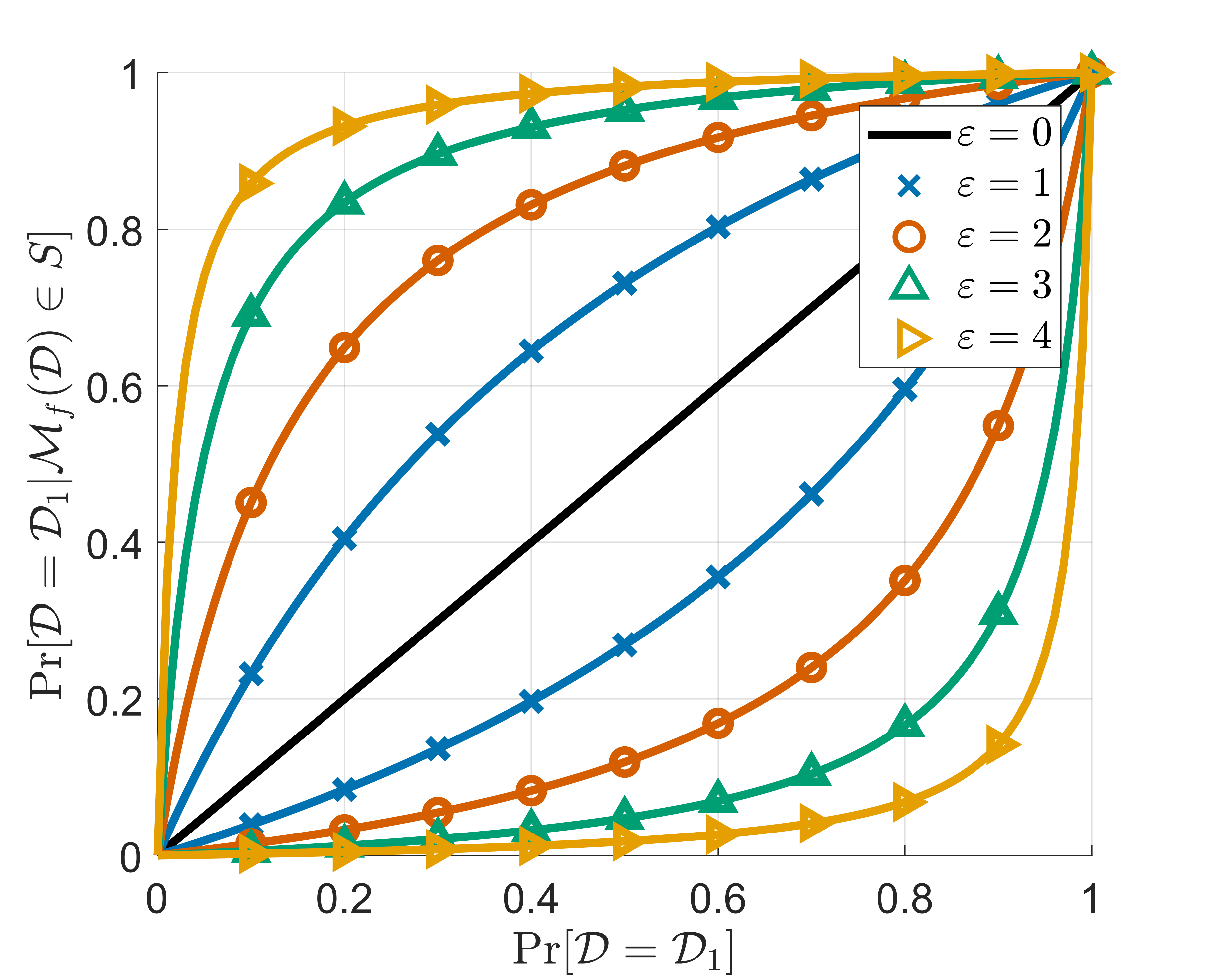}
  \caption{\label{fig:decision-theory} Pure-Differential Privacy limits the amount of information an adversary can gain from the outcome of private query. Based on the adversary's initial estimate of the alternative hypothesis, $\Prr[\Dc = \Dc_1]$, a Differentially Private mechanism bounds the conditional probability given the outcome of the query. Each pair of matching curves represents the lower and upper bound for an adversary's estimate of the alternative hypothesis after observing the outcome of the privacy mechanism. As the privacy budget $\varepsilon$ is increased, the bound of the adversary's updated estimate is increased.}
\end{figure}
When a mechanism that only achieves approximate-Differential Privacy is deployed, these bounds become probabilistic.

\subsection{Common Privacy Mechanisms}
\label{sec:priv-mechs}
The choice of probability density from which the noise is drawn significantly impacts the level of privacy achieved. To prevent bias from being introduced into the query, all mechanisms considered here are chosen to have zero mean.

\cite{dwork06c} and \cite{roth14} introduce the Laplace mechanism, one of the most commonly used mechanisms, which samples noise from the Laplace density:
\begin{equation}
p_{\text{Lap}}(x) = \frac{1}{2b}e^{\frac{-|x|}{b}}.
\end{equation}
Here, $b$ determines the spread of the distribution. The Laplace mechanism has been shown to satisfy pure Differential Privacy, equation (\ref{eqn:pureDP}) \cite{dwork06c, roth14}. Unfortunately, the Laplace density does not trivially extend to local Differential Privacy, limiting its application in methods such as Federated Learning. This property additionally makes the Laplace mechanism challenging to use for training Neural Networks, which rely heavily on the repeated compositions of a mechanism, so it is not commonly employed for deep learning.

Another frequently employed mechanism is the Gaussian mechanism, studied in \cite{dwork06c} and \cite{roth14}. This mechanism injects noise, drawn from a normal density with a mean of zero, into the output of a query:
\begin{equation}
p_{\text{Gaus}}(x) = \frac{1}{\sigma\sqrt{2\pi}}e^{\frac{-1}{2}(\frac{x}{\sigma})^2}.
\end{equation}
Since the Gaussian density is closed under convolution, the mechanism naturally extends to environments that require local application of the mechanism. However, the Gaussian mechanism is only approximately Differentially Private, i.e., it requires $\delta > 0$ in equation (\ref{eqn:approxDP}). Traditionally, this has not been seen as an issue because, as shown in \cite{roth14}, the repeated composition of approximate-DP methods scales better than simple composition pf pure-DP methods.

The Exponential mechanism is another noteworthy approach introduced by \cite{mcsherry07}. This mechanism selects outputs from a set probabilistically, weighting them according to their utility scores. By carefully choosing the scoring function, it provides a way to balance privacy and utility effectively. The noise added by the Exponential mechanism is drawn from the exponential density:
\begin{equation}
p_{\text{Exp}}(x) = \lambda e^{-\lambda x}.
\end{equation}
Here, $\lambda$ controls the rate of decay of the distribution. However, designing appropriate scoring functions that accurately capture utility while ensuring privacy remains a significant challenge in practical implementations.

With these common mechanisms in mind, we next proceed to define the Symmetric alpha-Stable mechanism and present novel analysis of its properties.

\section{The Symmetric alpha-Stable Mechanism}
\label{sec:sas}
To begin, it is essential to note that the Gaussian density belongs to a broader family of distributions called the L\'evy alpha-Stable densities, all of which exhibit closure under convolutions; see \cite{levy25}. However, it is shown by \cite{dwork06b} and \cite{roth14} that, in the realm of Differential Privacy, the Gaussian mechanism only adheres to condition (\ref{eqn:approxDP}), approximate Differential Privacy. This section delves into the characteristics of a privacy mechanism drawn from a subset of the L\'evy alpha-stable family. We refer to such amechanisms as Symmetric alpha-Stable mechanisms and provide a proof that they adhere to condition (\ref{eqn:pureDP}), pure Differential Privacy.

The concept of stable densities, extensively explored in \cite{levy25}, refers to a particular set of probability distributions. These distributions possess a notable property; the convolution of two distributions from the family is also a member of the family; this is otherwise known as closure under convolution.

\begin{definition} (The Stable Family)
    A probability density function $Y$ is termed \textit{stable} if, for any constants $a, b > 0$, there exist constants $ c(a,b) > 0$ and $d(a,b) \in \mathbb{R}$ such that the following holds for two independent and identically distributed random variables $Y_1$ and $Y_2$:
    \begin{equation}
        aY_1 + bY_2 = cY+d.
    \end{equation}
    If $d$ equals 0, the distribution is termed \textit{strictly stable}.
    \vspace{-2mm}
    \begin{equation*}\tag*{\textrm{$\blacktriangleleft$}}\end{equation*}
\end{definition}

\cite{nolan20} shows that, except for certain special cases, there is no known closed-form expression for the density of a general stable distribution. Nonetheless, several parameterizations of the characteristic function of a stable density are documented; see \cite{nolan20}. One common representation of the characteristic function is as follows:
\begin{equation}
\label{eqn:char}
    \varphi(t; \alpha, \beta, \gamma, \mu) = \exp({it\mu - |\gamma t|^\alpha + i \beta \textrm{sgn}(t) \Phi(t)}),
\end{equation}
where
\begin{equation}
    \Phi(t) = \begin{cases}
        \tan(\frac{\pi \alpha}{2}) & \alpha \neq 1 \\
        -\frac{2}{\pi}\log|t| & \alpha = 1.
    \end{cases}
\end{equation}
The density function is then given by the integral:
\begin{equation}
    \label{eqn:stab}
    p(x; \alpha, \beta, \gamma, \mu) = \frac{1}{2\pi}\int_{-\infty}^{\infty}\varphi(t; \alpha, \beta, \gamma, \mu)e^{-ixt}dt.
\end{equation}

In Figure \ref{fig:sas_density}, we present three examples of the symmetric form $\beta=0$: $\alpha=1$ (blue), $\alpha=1.5$ (orange), and $\alpha=2$ (green). Each of the three depicted densities has zero for the location parameter ($\mu=0$) and unit scale ($\gamma=1$).
\begin{figure}[ht!]
	\centering
	\includegraphics[width=0.6\linewidth]{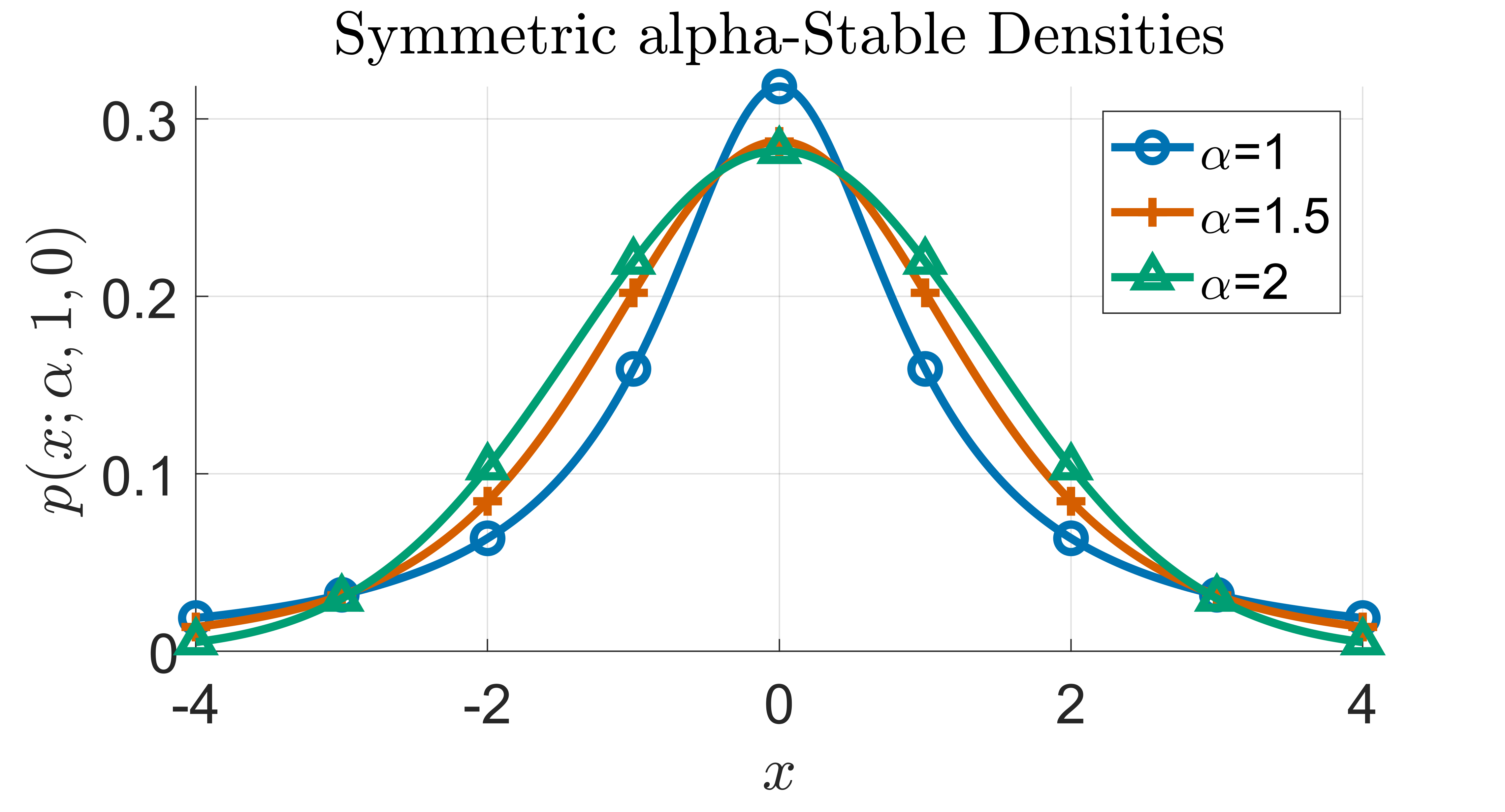}
	\caption{\label{fig:sas_density} The family of Symmetric alpha-Stable densities consists of bell shaped densities with varying tail weights determined by the stability parameter $\alpha$. This family of densities is unique because it is the only set of densities that are closed under convolution. When $\alpha=1$, shown in {\bf \color{blue} blue $\circ$}, the density is known as the Cauchy. When $\alpha=2$, shown in {\bf \color{ForestGreen} green $\triangle$}, the density is known as the Gaussian. No other values of alpha (for the symmetric case) have a known closed form solution, for example $\alpha=1.5$, shown in {\bf \color{Bittersweet} orange +}.}
\end{figure}
The Symmetric alpha-Stable densities with $\alpha =1$ and $\alpha=2$ are the only two densities with support on the whole real line that have a known closed form. When $\alpha = 1$, the density is known as the Cauchy density. When $\alpha = 2$, we recover the Gaussian density.

Working with the family of stable distributions presents a significant challenge due to the absence of a closed-form solution for the general density. This difficulty arises because the value at any given point is determined by integrating an infinitely oscillating function.

Denote the real part of the integrand in Equation (\ref{eqn:stab}) by $q(t;x)$. A visualization of this function for $x=10$ is illustrated in Figure \ref{fig:slice_sas_density}.
\begin{figure}[ht!]
  \centering
  \includegraphics[width=0.6\linewidth]{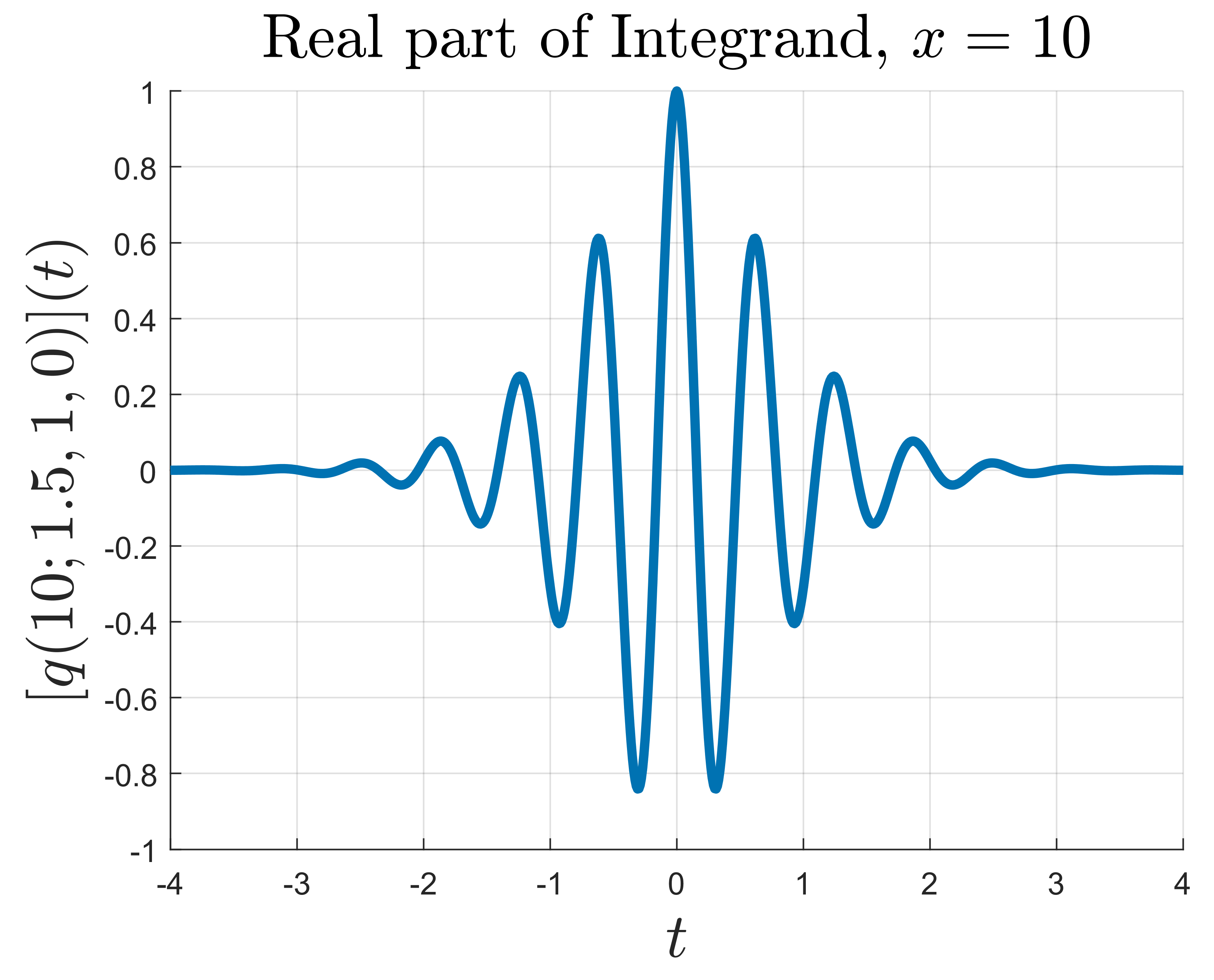}
  \caption{\label{fig:slice_sas_density} The real part of the integrand of (\ref{eqn:stab}) for $\alpha=1.5$, $\gamma=1$, and $\mu=0$ is an infinitely oscillating function. The value of the stable density with these parameters at the point $x=10$ is the integral of this function on the real line.}
\end{figure}

In order for Equation (\ref{eqn:stab}) to represent a valid probability density, the parameter $\alpha$ must fall within the interval $(0, 2]$. The value of $\alpha$ dictates the rate of decay of the density's tail. The expected value of the density is only defined in the range $\alpha \in (1, 2]$ and is not defined for $\alpha \leq 1$. Furthermore, the density exhibits infinite variance for $\alpha \in (0, 2)$ and finite variance only when $\alpha = 2$. In this study, we confine $\alpha$ to $(1, 2]$, deferring median or mode estimators for future exploration. The parameter $\beta$, constrained to $(-1, 1)$, serves as a measure linked to skewness, noting that the strict definition of skewness lacks meaning for $\alpha < 2$. Our focus lies specifically on symmetric alpha-stable (SaS) densities, where $\beta = 0$:

\begin{equation}
  \label{eqn:sas}
  \begin{aligned}
    p_{SaS}(x;\alpha, \gamma, \mu) & := \\ p(x;\alpha, 0, \gamma, \mu) & = \frac{1}{2\pi}\int_{-\infty}^{\infty}e^{-|\gamma t|^\alpha - it(x-\mu)}dt.
  \end{aligned}
\end{equation}
SaS densities have a known closed form for two values of the parameter $\alpha$: the Cauchy, for $\alpha = 1$, and the Gaussian, for $\alpha=2$. The last two parameters, $\gamma > 0$ and $\mu\in\mathbb{R}$, are the scale and location parameter's respectively. 
\begin{remark}
    For stable densities, it is common for the location parameter to be denoted $\delta$ rather than $\mu$, to signify that it is not always equal to the expected value. In our context, we choose to use $\mu$ and reserve $\delta$ for the definition of approximate Differential Privacy (\ref{eqn:approxDP}) as is common in the DP literature. Because we are restricting the domain of interest to $\alpha \in (1,2]$, where the mean is well-defined, we do not believe this notation will be cause for confusion.
\end{remark}

We are now prepared to define the Symmetric alpha-Stable mechanism:
\begin{definition}
\label{def:sas}
    (The Symmetric alpha-Stable mechanism)
    For a given dataset $\mathcal{D}$ and a query function $f$, we define a privacy mechanism $\mathcal{M}_f$ to be a Symmetric alpha-Stable (SaS) mechanism if each element of the vector of injected values, $Y_i$ for $i\in\{1,...,m\}$, is drawn independently from a SaS density
    
    \begin{equation}
      \begin{aligned}
        p_{SaS}(x;\alpha,\gamma) := & \\ p(x; \alpha, 0, \gamma, 0) = & \frac{1}{2\pi}\int_{-\infty}^{\infty}e^{-|\gamma t|^\alpha - itx}dt.
        \end{aligned}
    \end{equation}
    \vspace{-6mm}
    \begin{equation*}\tag*{\textrm{$\blacktriangleleft$}}\end{equation*}
\end{definition}
While this family of mechanisms is closely related to the Gaussian, the Gaussian mechanism only satisfies approximate Differential Privacy. We show in the next section that the heaviness of the SaS density's tail allows the privacy mechanism to satisfy pure-Differential Privacy (when $\alpha < 2$).

\section{Pure-Differential Privacy of SaS Mechanism}
\label{sec:priv}
In this section, we establish that the SaS mechanism, when $\alpha\in[1,2)$, satisfies (\ref{eqn:pureDP}), providing $\varepsilon$-Differential Privacy. A significant challenge in working with stable densities, excluding the Cauchy and Gaussian, arises from their lack of a closed-form expression for the density. To ensure that the stable distribution covers the entire real number line, it's crucial to demonstrate that the privacy loss remains finite within a compact set. To prevent the denominator of the privacy loss in equation (\ref{eqn:privloss}) from becoming zero, resulting in infinite privacy loss, we first provide a lemma that states that the density is nonzero over the whole real line.
\begin{lemma}
  \label{lem:sup}
  (Support of SaS Density)
  The support of the symmetric alpha-stable density (\ref{eqn:sas}) is $\mathbb{R}$.
\end{lemma}
\begin{proof}
  See \cite[Lemma 1.1]{nolan20}.
\end{proof}

Additionally, we recall a partial sum expansion, as described in \cite{bergstrom52}, wherein the remainder term possesses a smaller order of magnitude (for large $|x|$) than the final term in the series.
\begin{lemma}
  \label{lem:sum}
  (Finite Series Expansion of SaS Distribution)
  The symmetric alpha-stable density (\ref{eqn:sas}), with $\alpha\in(1,2]$ and $\gamma=1$, admits the following finite series expansion:
  \begin{equation}
    \label{eqn:series}
    \begin{aligned}
      & p_{SaS}(x; \alpha, 1, 0) = \\ & - \frac{1}{\pi} \sum_{k=1}^n (-1)^k \frac{\Gamma(\alpha k+1)}{(x)^{\alpha k + 1}} \sin \bigg( \frac{k \alpha \pi}{2} \bigg) + O\bigg(x^{-\alpha(n+1)-1}\bigg),
    \end{aligned}
  \end{equation}
  as $|x| \to \infty$.
\end{lemma}
\begin{proof}
  \cite{bergstrom52} offers an expanded form of (\ref{eqn:series}) that is valid for the complete range $\beta \in (-1,1)$. However, because we have restricted the parameter set to $\beta = 0$, we only require the form provided for our purposes.
\end{proof}

We use the foregoing lemma to argue that the privacy loss remains bounded as the observation $|x|$ tends to infinity. However, Eq. (\ref{eqn:series}) is stated for $\gamma=1$. The next lemma states that the asymptotic behavior of the privacy loss as $|x|\to\infty$ is independent of $\gamma$.

\begin{lemma}
\label{lem:scale}
(No Scale Dependence in the Limit)
Let $\mathcal{D}_1 \simeq \mathcal{D}_2$ be two neighboring datasets. Denote by $\mathcal{L}^{SaS}_{\mathcal{D}_1|| \mathcal{D}_2}(x; \gamma)$ the privacy loss of observation $x$ for a bounded query $f$ perturbed by a SaS mechanism $\mathcal{M}_f$ with scale parameter $\gamma$. In the limit as $|x|$ tends to $\infty$, the behavior of the privacy loss is indistinguishably asymptotic for distinct choices of $\gamma$:
\begin{equation}
    \lim_{|x| \to \infty} \mathcal{L}^{SaS}_{\mathcal{D}_1|| \mathcal{D}_2}(x; \gamma_1) = \lim_{|x| \to \infty} \mathcal{L}^{SaS}_{\mathcal{D}_1|| \mathcal{D}_2} (x; \gamma_2),
\end{equation}
for $\gamma_1 \neq \gamma_2$.
\end{lemma}
\begin{proof}
% The proof follows directly the limit as $|x|$ is taken to $\infty$ in equation (\ref{eqn:sas}) with the substitutions $\hat{t} = \gamma t$ and $\hat{x} = x / \gamma$.
  \begin{equation}
    \label{eqn:gamma_1}
    \begin{aligned}       
      p(x; \alpha, \gamma, \mu) & =  \frac{1}{2\gamma\pi}\int_{-\infty}^{\infty} e^{-| \hat{t}|^\alpha-i(\hat{x}-\mu_i)\hat{t}}d\hat{t}\\
                                & = p(\hat{x}; \alpha, 1, \mu).
    \end{aligned}
  \end{equation}
  Substituting (\ref{eqn:gamma_1}) into the privacy loss function (\ref{eqn:plf}) gives
  \begin{equation}
    \begin{aligned}       
      \mathcal{L}^{SaS}_{\mathcal{D}_1|| \mathcal{D}_2}(x; \gamma) & = \ln \frac{\int_{-\infty}^{\infty} e^{-| \hat{t}|^\alpha-i(\hat{x}-f(\mathcal{D}_1))\hat{t}}d\hat{t}}{\int_{-\infty}^{\infty} e^{-| \hat{t}|^\alpha-i(\hat{x}-f(\mathcal{D}_2))\hat{t}}d\hat{t}} \\
                                                                   & = \mathcal{L}^{SaS}_{\mathcal{D}_1|| \mathcal{D}_2}(\hat{x}; 1).
    \end{aligned}
  \end{equation}
  Observing that $|\hat{x}|$ tends to $\infty$ as $|x|$ is driven to $\infty$, we have
  \begin{equation}
    \lim_{|x|\to\infty} \mathcal{L}^{SaS}_{\mathcal{D}_1|| \mathcal{D}_2}(x; \gamma) = \lim_{|\hat{x}|\to\infty} \mathcal{L}^{SaS}_{\mathcal{D}_1|| \mathcal{D}_2}(\hat{x}; 1), \quad \forall \gamma.
  \end{equation}
  In the limit, as $|x|$ tends to infinity, the shift and scale of $\hat{x}_1$ and $\hat{x}_2$ are irrelevant.
\end{proof}

With the results above, we are now in a position to state and prove our main contribution, namely, that for $\alpha \in [1,2)$, the privacy loss of the SaS mechanism is bounded, i.e. the SaS mechanism is $\varepsilon$-differentially private.
\begin{theorem}
  \label{thm:sas-dp}
  (The SaS mechanism is $\epsilon$-DP)  
  Let $\mathcal{D}_1 \simeq \mathcal{D}_2$ be two neighboring datasets and let $f$ be a bounded query that operates on them. Consider the SaS mechanism, which we denote by $\mathcal{M}_f$, with stability parameter $\alpha$ in the reduced range  $\alpha \in [1, 2)$. Then, the mechanism $\mathcal{M}_f$ satisfies (\ref{eqn:pureDP}), pure Differential Privacy.
\end{theorem}
\begin{proof}
  Each element of the mechanism's output is the perturbation of the query's response by an independent sample from the uni-variate density in (\ref{eqn:sas}). Thus, the joint density is equal to the product of the individual densities. As a result, we can express the privacy loss for a given observation vector $\x$ as 
  \begin{equation}
    % \label{eqn:sas-plf}
    \mathcal{L}^{SaS}_{\mathcal{D}_1|| \mathcal{D}_2}(\x) = \ln \frac{\prod\limits_{i = 1}^m p_{SaS}(x_i; \alpha, \gamma,f(\mathcal{D}_1)_i)}{\prod\limits_{i = 1}^mp_{SaS}(x_i; \alpha, \gamma,f(\mathcal{D}_2)_i)}.
  \end{equation}
  This can be written as the sum of the log-ratios of the individual elements:
  \begin{equation}
    \mathcal{L}^{SaS}_{\mathcal{D}_1, \mathcal{D}_2}(\x) = \sum_{i=1}^m \ln \frac{p_{SaS}(x_i; \alpha, \gamma,f(\mathcal{D}_1)_i)}{p_{SaS}(x_i; \alpha, \gamma,f(\mathcal{D}_2)_i)}.
  \end{equation}
  Without loss of generality, let this sum be written in decreasing order of magnitudes of the terms, i.e. the first term, $i=1$, has the largest magnitude. We now have the following bound:
  \begin{equation}
    \label{eqn:sas-plf}
    \big|\mathcal{L}^{SaS}_{\mathcal{D}_1|| \mathcal{D}_2}(\x)\big| \leq m \Big|\ln \frac{p_{SaS}(x_1; \alpha, \gamma,f(\mathcal{D}_1)_1)}{p_{SaS}(x_1; \alpha, \gamma,f(\mathcal{D}_2)_1)}\Big|.
  \end{equation}
  Our objective is to prove that the right side of (\ref{eqn:sas-plf}) is bounded as function of $x_1$, which will imply, by Theorem \ref{thm:atoms}, that the mechanism is $\varepsilon$-differentially private. We do so by first proving that the privacy loss is bounded on any compact set. Note that this is not immediate, since we are dealing with the log of a ratio and have no assurance that the numerator or denominator ever vanishes. Then, we show that in the limit as $|x|$ tends to infinity, the privacy loss tends to $0$, and thus does not diverge.
  
  Initially, let $x_1$ be an element in a compact set $[a, b] \subset \mathbb{R}$. The log-ratio of the densities could become unbounded within a finite interval in two ways: the argument vanishes or diverges. Consider first the case where one of the densities vanishes within the interval. By Lemma \ref{lem:sup}, an SaS density has support on the entire real line, $\mathbb{R}$. Therefore, the density is strictly positive over all compact sets $[a,b] \subset \mathbb{R}$. 
  
  Then, we consider if the numerator or denominator of (\ref{eqn:sas-plf}) could be unbounded within the interval $[a,b]$. For simplicity, let $\mu=0$ and apply the substitution $e^{-ix_1} = \cos(tx_1) - i\sin(tx_1)$ to the representation of the SaS density (\ref{eqn:sas}):
  \begin{equation}
    p_{SaS}(x_1;\alpha, \gamma, 0) = \frac{1}{2\pi}\int_{-\infty}^{\infty}e^{-|\gamma t|^\alpha}(\cos(tx_1) - i\sin(tx_1))dt.
  \end{equation}
  Splitting the integral we have
  \begin{equation}
    \begin{aligned}    
      p_{SaS}(x_1;\alpha, \gamma, 0) = & \frac{1}{2\pi}\int_{-\infty}^{\infty}e^{-|\gamma t|^\alpha}\cos(tx_1)dt \\ - & i \frac{1}{2\pi}\int_{-\infty}^{\infty}e^{-|\gamma t|^\alpha} \sin(tx_1)dt.
    \end{aligned}
  \end{equation}
  Since $\sin(tx_1)$ is an odd function the second integral vanishes:
  \begin{equation}
    \label{eqn:cos-form}
    p_{SaS}(x_1;\alpha, \gamma, 0) = \frac{1}{2\pi}\int_{-\infty}^{\infty}e^{-|\gamma t|^\alpha}\cos(tx_1)dt.
  \end{equation}
  As $\cos(tx_1)$ is bounded above by $1$, the density is bounded above:
  \begin{equation}
    \label{eqn:a_part}
    p_{SaS}(x_1;\alpha, \gamma, 0) \leq   \frac{1}{2\pi}\int_{-\infty}^{\infty}e^{-|\gamma t|^\alpha}dt.
  \end{equation}
  Observe that the integrand in (\ref{eqn:a_part}) is symmetric about $t=0$, so we can remove the absolute value by adjusting the limits of integration:
  \begin{equation}
    \label{eqn:30}
    p_{SaS}(x_1;\alpha, \gamma, 0) \leq   \frac{1}{\pi}\int_{0}^{\infty}e^{-(\gamma t)^\alpha}dt.
  \end{equation}
  Letting $\hat{t} = (\gamma t)^\alpha$, substitute $\hat{t}$ into the inequality (\ref{eqn:30}):
  \begin{equation}
    \label{eqn:31}
    \begin{aligned}
      p_{SaS}(x_1;\alpha, \gamma, 0) & \leq   \frac{1}{\alpha\gamma\pi}\int_{0}^{\infty}\hat{t}^{\frac{1}{\alpha}-1}e^{-\hat{t}}d\hat{t}\\ 
                                     & = \frac{\Gamma(\frac{1}{\alpha})}{\alpha \gamma\pi},
    \end{aligned}
  \end{equation}
  where $\Gamma$ is the standard Gamma function. Note that the Gamma function is finite on the interval $1/\alpha \in (1/2, 1)$; see \cite{oeis}. Equation (\ref{eqn:31}) states that the density $p_{SaS}$ is bounded over the real line. It is therefore bounded on the compact subset $[a,b]$. We proceed to prove that the privacy loss remains bounded in the limit as $|x_1|$ tends to infinity. 
  
  Recall the series expansion presented in Lemma \ref{lem:sum}, for scale $\gamma=1$. Truncate the series to a single term by taking $n=1$ and consider the privacy loss after substitution in (\ref{eqn:sas-plf}):
  \begin{equation}
    \big|\mathcal{L}^{SaS}_{\mathcal{D}_1||\mathcal{D}_2}(\x)\big| \leq m\Big|\ln\frac{\big(x_1 - f(\mathcal{D}_1)\big)^{-\alpha-1} + O(x_1^{-2\alpha - 1})}{\big(x_1 - f(\mathcal{D}_2)\big)^{-\alpha-1}  + O(x_1^{-2\alpha - 1})}\Big|.
  \end{equation}
  In the limit, as $|x_1|$ tends infinity, the error terms in the numerator and denominator are dominated by the first terms:
  \begin{equation}
    \begin{aligned}
      \lim_{||\x||\to\infty} & \big|\mathcal{L}^{SaS}_{\mathcal{D}_1||\mathcal{D}_2}(\x)\big| \leq \\ \lim_{|x_1|\to\infty} & m\Big|\ln\frac{\big(x_1 - f(\mathcal{D}_1)\big)^{-\alpha-1} + O(x_1^{-2\alpha - 1})}{\big(x_1 - f(\mathcal{D}_2)\big)^{-\alpha-1}  + O(x_1^{-2\alpha - 1})}\Big| = \\
      \lim_{|x_1|\to\infty} & m\Big|\ln\frac{\big(x_1 - f(\mathcal{D}_1)\big)^{-\alpha-1}}{\big(x_1 - f(\mathcal{D}_2)\big)^{-\alpha-1}}\Big| = 0,
    \end{aligned}
  \end{equation}
   and the privacy loss converges to $0$. By Lemma \ref{lem:scale}, the choice of $\gamma$ does not impact the asymptotic behavior. Since this result holds for any value of $\x \in \mathcal{R}(\mathcal{M}_f)$, by Theorem \ref{thm:atoms}, we have proved that the SaS mechanism is $\varepsilon$-differentially private.
\end{proof}

Although Theorem \ref{thm:sas-dp} establishes that the SaS mechanism is purely-Differentially Private, it does not offer a connection between the sensitivity of the query $\Delta f$, the scale of the noise distribution $\gamma$, and the achieved level of privacy $\varepsilon$. This limitation stems from the absence of a known closed-form expression for the density $p_{SaS}$. Before pursing further details on these relationships, we revisit the Differential Privacy characteristics of two widely used privacy mechanisms to facilitate the subsequent comparison.

\section{Privacy Scaling with Noise}
\label{sec:scale}
In this section, we recall the characteristics of two common privacy mechanisms put forth in \cite{dwork06, roth14}: the Laplace mechanism and the Gaussian mechanism. After discussing these mechanisms, we proceed to study the relation between the privacy budget $\varepsilon$ and the scale $\gamma$ of the SaS mechanism and to provide related numerical results. We proceed to argue that the privacy budget of the SaS mechanism scales with the same order as the Laplace and Gaussian mechanisms, i.e., we wish to show that
\begin{equation}
\label{eqn:scaling}
        \varepsilon_{SaS} \stackrel{?}{\propto} \frac{\Delta_1 f}{\gamma}
\end{equation}
which is similar to 
\begin{equation}
\label{eqn:other-scaling}
    \begin{aligned}
        \varepsilon_{Lap} = \frac{\Delta_1 f}{b} \quad \textrm{and} \quad \varepsilon_{Gau}  \propto \frac{\Delta_2 f}{\sigma}.
            \end{aligned}
\end{equation}

\subsection{Level of privacy afforded by the SaS mechanism}
For a given problem, there are three factors to consider when setting the parameters of a mechanism: the sensitivity of the query $\Delta f$, the scale of the noise $\gamma$, and the privacy budget $\varepsilon$. In this section, we study the relationship between these three quantities for the SaS mechanism. 

Theorem \ref{thm:sas-dp} bounds the privacy loss by considering the largest component of the $m$-dimensional response of a query (see Eq. (\ref{eqn:sas-plf})). This motivates us to focus on the case $m=1$, i.e., we now restrict to real-valued queries, $f(\mathcal{D}) \in [a, b] \subset \mathbb{R}$. Furthermore, in this section, when referring to the sensitivity of query $f$ we exclusively use the $\ell_1$-sensitivity and denote it by $\Delta_1$.

We begin by establishing the linear relation between sensitivity and scale. To do so, we first prove that the extremes of the privacy loss, for a given privacy budget $\varepsilon$, occur when the query over datasets $\mathcal{D}_1 \simeq \mathcal{D}_2$ returns values in the boundary of the range, $\mathcal{R}(f)$. For instance, when $f(\mathcal{D}_1) = b$ and $f(\mathcal{D}_2) = a$ (or vice versa, $f(\mathcal{D}_1)=a$ and $f(\mathcal{D}_2)=b$, by symmetry of the absolute value of the query).

In order to prove that the privacy loss is maximized at the boundary of the query's range, we first establish that the density is monotonic on each semi-infinite interval to the left and right of the location parameter $\mu$. We give a proof for the generic symmetric stable density using the fact that the density is \textit{bell-shaped}, the definition of which is recalled next from \cite{kwasnicki20}.

\begin{definition}
  \label{def:bell}
  (Bell-Shaped Function)
  A continuous real-valued function is said to be \textit{bell-shaped} if the $n^{th}$ derivative, $f^{(n)}$ for each $n \in \mathbb{N}_0$, changes sign exactly $n$ times over its domain.
   \vspace{-2mm}
  \begin{equation*}\tag*{\textrm{$\blacktriangleleft$}}\end{equation*}
\end{definition}

\begin{lemma}
  \label{lem:bell}
  (Monotonic First Derivative)
  The symmetric alpha-stable density (\ref{eqn:sas}) with location parameter $\mu$ is monotonically increasing from $-\infty$ to $\mu$ and monotonically decreasing from $\mu$ to $\infty$. 
\end{lemma}
\begin{proof}
  See the proof of \cite[Cor. 1.3]{kwasnicki20} which asserts that all stable distributions are bell-shaped densities. Taking $n=1$ in Def. \ref{def:bell} implies that the first derivative of the density, $f'$, changes sign exactly once. Because the density is symmetric, the change in sign must occur at the axis of symmetry and the density must then decrease monotonically to $0$ in the limit as $|x| \to \infty$.
\end{proof}

We now utilize Lemma \ref{lem:bell} to prove that, out of all neighboring datasets $\mathcal{D}_1 \simeq \mathcal{D}_2$, the maximum of the privacy loss occurs at the boundary of the query's range $[a,b]$. Recall that the SaS mechanism involves injecting noise with a location parameter of $0$. Thus, the location parameter is the result of the query, $\mu_i = f(\mathcal{D}_i)$, and is itself bounded by the range of the query. By Theorem \ref{thm:sas-dp}, the privacy loss of the SaS mechanism is bounded. As a result, we denote by $x^*(\mu_1,\mu_2)$ the point at which the maximum privacy loss occurs as a function of the location parameters $\mu_1$ and $\mu_2$ generated by datasets $\mathcal{D}_1$ and $\mathcal{D}_2$ respectively.

\begin{theorem}
  \label{thm:max}
  (Privacy Loss Maximization Occurs at Boundary)
  Let $\mathcal{D}_1 \simeq \mathcal{D}_2$ be neighboring datasets and denote by $f$ a bounded query that operates on them and returns values in the compact set $[a, b] \subset \mathbb{R}$. Denote the SaS mechanism's privacy loss for an observation $x$ by $\mathcal{L}^{SaS}_{\mathcal{D}_1||\mathcal{D}_2}(x)$. Let the location parameters of the two densities be $\mu_1 = f(\mathcal{D}_1)$ and $\mu_2 = f(\mathcal{D}_2)$, with $\mu_1 \neq \mu_2$. Then
  \begin{equation}
    \mathcal{L}^{SaS}_{\mathcal{D}_1|| \mathcal{D}_2}\Big(x^*(\mu_1,\mu_2)\Big) \leq \mathcal{L}^{SaS}_{\mathcal{D}_1|| \mathcal{D}_2}\Big(x^*(b,a)\Big).
  \end{equation}
  \end{theorem}
\begin{proof}
  Without loss of generality, take $\mu_1 > \mu_2$. Recall that the privacy loss, (\ref{eqn:plf}), is given by the log-ratio of two densities. Consider Figure \ref{fig:max_bounds} and let $p(x;\mu_1)$, in blue, and $p(x;\mu_2)$, in orange, represent the numerator and denominator of the privacy loss respectively. Let $\epsilon$ be a value in $[0, b-\mu_1]$.
  \begin{figure}[ht!]
    \centering
    \includegraphics[width=0.8\linewidth]{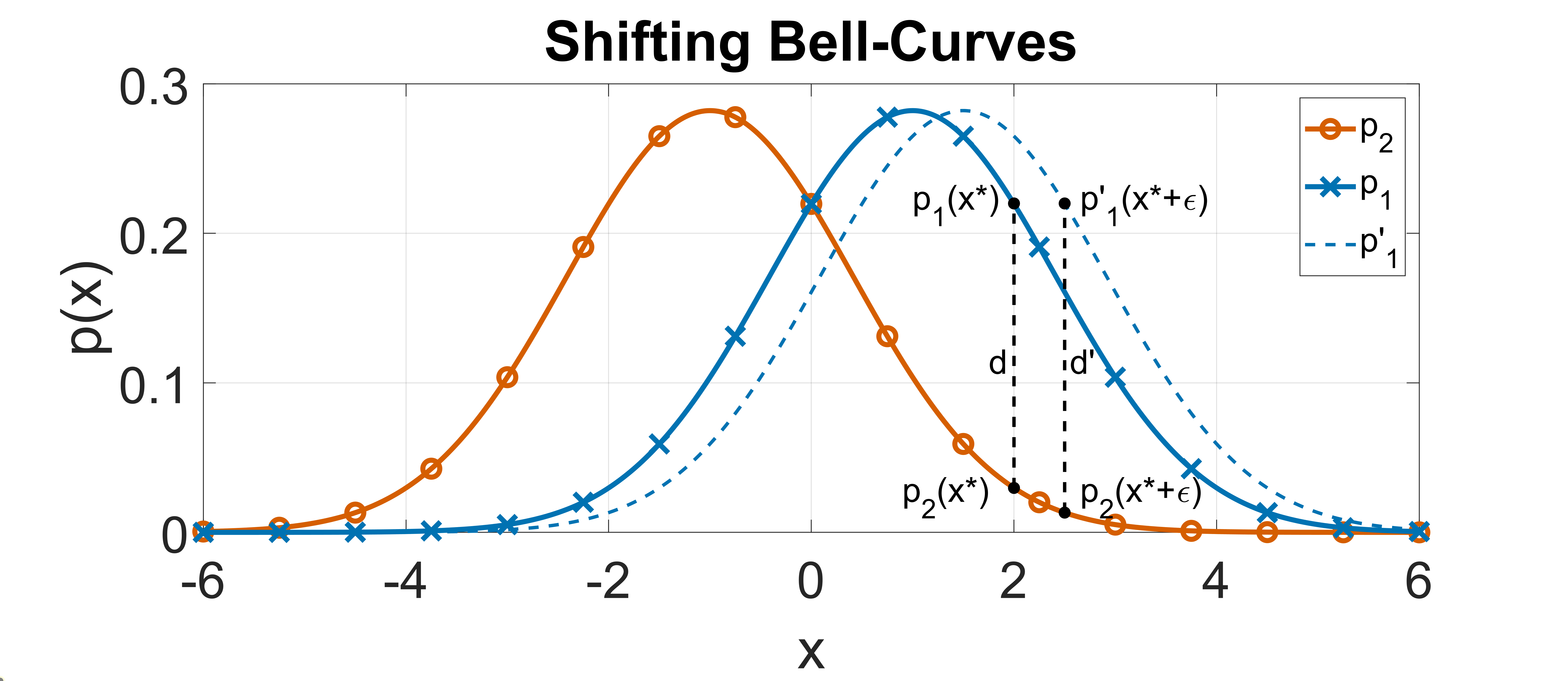}
    \caption{\label{fig:max_bounds} Consider two bell curves, shown here as {\bf \color{blue} $p_1$ in blue} and {\bf \color{Bittersweet} $p_2$ in orange}, with location parameters $\mu_1 > \mu_2$ respectively. Given a point $x^* > \mu_1$, denote by $d$ the distance between the curves at $x^*$: $d := p_1(x^*) - p_2(x^*)$. Shifting the distribution $p_1$ to the right by some positive value $\epsilon$, gives the curve $p'_1$ shown as a dotted line. By Lemma \ref{lem:bell}, the distance $d' := p'_1(x^*+\epsilon) - p_2(x^*+\epsilon)$ is necessarily larger than $d$.}
  \end{figure}
  First, we show that if the privacy loss achieves a maximum $x^*(\mu_1, \mu_2)$, then $\mu_1 \leq x^*(\mu_1, \mu_2)$. Observe that, by construction, $p(x=\mu_1;\mu_1) \geq p(x=\mu_1;\mu_2)$. Consider a point to the left of $\mu_1$. By the symmetry of SaS densities, $p(\mu_1-\epsilon;\mu_1) = p(\mu_1+\epsilon;\mu_1)$ and because the first derivative is negative, Lemma \ref{lem:bell}, $p(\mu_1-\epsilon;\mu_2) \geq p(\mu_1+\epsilon;\mu_2)$. Thus,
  \begin{equation}
    \mathcal{L}^{SaS}_{\mathcal{D}_1|| \mathcal{D}_2}(\mu_1-\epsilon) \leq \mathcal{L}^{SaS}_{\mathcal{D}_1|| \mathcal{D}_2}(\mu_1+\epsilon).
  \end{equation}
  
  Next, let $\mu_1 < b$. Then, observe that $p(x^*(\mu_1,\mu_2);\mu_1) = p(x^*(\mu_1,\mu_2)+\epsilon;\mu_1+\epsilon)$, illustrated by the upper two marked points in Figure \ref{fig:max_bounds}. Similarly, by Lemma \ref{lem:bell}, $p(x^*(\mu_1,\mu_2)+\epsilon;\mu_2) \leq p(x^*(\mu_1,\mu_2);\mu_2)$, marked by the two lower points. Thus, $\mathcal{L}$ can only be made larger by increasing $\mu_1$ in the direction of the bound $b$. Likewise, if $\mu_1 = b$, then shifting the distribution to the left can only decrease the maximum. A similar argument shows that the log-ratio cannot be decreased by shifting $p(x;\mu_2)$ towards $p(x;a)$, which completes the proof.
\end{proof}

Because the maximum of the privacy loss is invariant to translation, Theorem \ref{thm:max} additionally implies the following corollary. 
\begin{corollary} (Relative Location Parameter)
  Let $\mathcal{D}_1$ and $\mathcal{D}_2$ be any neighboring datasets. Consider the privacy loss of the SaS mechanism, with $\alpha \in (1,2)$, for a one-dimensional query $f$, with bounded range $\mathcal{R}(f) = [a, b]$. Denote by $\Delta_1$ the $\ell_1$-sensitivity of $f$. Then, for a given $\alpha \in (1, 2)$ and scale $\gamma$,
  \begin{equation}
    \label{eqn:sas-delta}
    \max_{\mathcal{D}_1 \simeq \mathcal{D}_2}\max_{x \in \mathbb{R}} \mathcal{L}_{\mathcal{D}_1||\mathcal{D}_2}^{SaS}(x) = \max_{x \in \mathbb{R}} \ln \frac{p_{SaS}(x; \alpha, \gamma, \Delta_1)}{p_{SaS}(x;\alpha, \gamma, 0)}. 
  \end{equation}
\end{corollary}
\begin{proof}
  The result follows directly from Theorem \ref{thm:max} and observing that the maximum of the privacy loss is invariant under translation.
\end{proof}

We can now assert that there is a linear relation between the sensitivity of the query $\Delta_1$ and the scale of the density $\gamma$.

\begin{theorem}
  \label{thm:linscale}
  (Linearity of Scale and Query's Sensitivity) 
  Let $\mathcal{D}_1 \simeq \mathcal{D}_2$ be neighboring datasets and $f$ be a one-dimensional query with bounded range $\mathcal{R}(f) = [a, b]$. Denote by $\Delta_1$ the $\ell_1$-sensitivity of $f$. Let $p_{SaS}$ be the SaS density as described in equation (\ref{eqn:sas}). Then, the level of privacy $\varepsilon$ remains the same if the sensitivity $\Delta_1$ and the scale $\gamma$ are both scaled by the same constant $c > 0$:
  \begin{equation}
    \label{eqn:lin-bound-proof}
    \max_{x'\in\mathbb{R}}\ln \frac{p_{SaS}(x'; \alpha, c\gamma, c\Delta_1)}{p_{SaS}(x';\alpha, c\gamma, 0)} = \max_{x\in\mathbb{R}}\ln \frac{p_{SaS}(x; \alpha, \gamma, \Delta_1)}{p_{SaS}(x;\alpha, \gamma, 0)}.
  \end{equation}
\end{theorem}
\begin{proof}
  We proceed by contradiction. Denote by $x^*$ optimal argument on the right side of (\ref{eqn:lin-bound-proof}). Consider the left side of (\ref{eqn:lin-bound-proof}) in terms of the expression (\ref{eqn:sas}):
  \begin{equation}
    \label{eqn:lin-bound}
    \max_{x'\in\mathbb{R}} \ln \frac{\int\limits_{-\infty}^{\infty}e^{-|c\gamma t|^\alpha - it(x'-c\Delta_1)}dt}{\int\limits_{-\infty}^{\infty}e^{-|c\gamma t|^\alpha - itx'}dt} .
  \end{equation}
  The change of variables $\hat{t} = c t$ results in the equivalent expression
  \begin{equation}
    \label{eqn:lin-bound2}
    \max_{x'\in\mathbb{R}} \ln \frac{\int\limits_{-\infty}^{\infty}e^{-|\gamma \hat{t}|^\alpha - i\hat{t}(\frac{x'}{c}-\Delta_1)}d\hat{t}}{\int\limits_{-\infty}^{\infty}e^{-|\gamma \hat{t}|^\alpha - i\hat{t}\frac{x'}{c}}d\hat{t}}
  \end{equation}
  Denote by $x'^*$ the location of the maximum in (\ref{eqn:lin-bound2}) and assume that it is not equal to $cx^*$. This leads to the following contradiction 
  \begin{equation}
    \label{eqn:lin-bound3}
    \begin{aligned}
      & \max_{cx'\in\mathbb{R}} \ln \frac{\int\limits_{-\infty}^{\infty}e^{-|\gamma \hat{t}|^\alpha - i\hat{t}(x'-\Delta_1)}d\hat{t}}{\int\limits_{-\infty}^{\infty}e^{-|\gamma \hat{t}|^\alpha - i\hat{t}x'}d\hat{t}} \neq   \\      & \max_{x\in\mathbb{R}} \ln \frac{\int\limits_{-\infty}^{\infty}e^{-|\gamma \hat{t}|^\alpha - i\hat{t}(x-\Delta_1)}d\hat{t}}{\int\limits_{-\infty}^{\infty}e^{-|\gamma \hat{t}|^\alpha - i\hat{t}x}d\hat{t}}
    \end{aligned}
  \end{equation}
  which is equivalent to
  \begin{equation}
    \label{eqn:lin-bound-proof-end}
    \max_{cx'\in\mathbb{R}}\ln \frac{p_{SaS}(x'; \alpha, \gamma, \Delta_1)}{p_{SaS}(x';\alpha, \gamma, 0)} \neq \max_{x\in\mathbb{R}}\ln \frac{p_{SaS}(x; \alpha, \gamma, \Delta_1)}{p_{SaS}(x;\alpha, \gamma, 0)}.
  \end{equation}
\end{proof}

\begin{remark} (Normalized Form)
\label{cor:simple}
Because the scale and sensitivity are related linearly, we can combine $\gamma$ and $\Delta_1$ into a single parameter $\hat{\gamma} = \gamma/\Delta_1$ by taking $c=1/\Delta_1$:
\begin{equation}
  \label{eqn:simple}
  \max_{x \in \mathbb{R}} \ln \frac{p_{SaS}(x; \alpha, \gamma, \Delta_1)}{p_{SaS}(x;\alpha, \gamma, 0)} = \max_{x' \in \mathbb{R}} \ln \frac{p_{SaS}(x'; \alpha, \hat{\gamma}, 1)}{p_{SaS}(x';\alpha, \hat{\gamma}, 0)}.
\end{equation}
\end{remark}

We use the normalized form on the right side of Eq. (\ref{eqn:simple}) to gain an intuitive understanding of how the maximum of the privacy loss behaves as $\alpha$ and $\gamma$ are allowed to vary. Figure \ref{fig:alpha_range} fixes $\gamma = \Delta_1 = 1$ and illustrates how the privacy loss approaches a straight line as $\alpha$ tends to $2$. Note that when $\alpha=2$, corresponding to the privacy loss of the Gaussian mechanism, the loss is unbounded, illustrating that the Gaussian mechanism is not purely differentially private.
% See Lemma \ref{prop:gaus} in Appendix \ref{appendix:A} for details. 
\begin{figure}[ht!]
  \centering
  \includegraphics[width=0.6\linewidth]{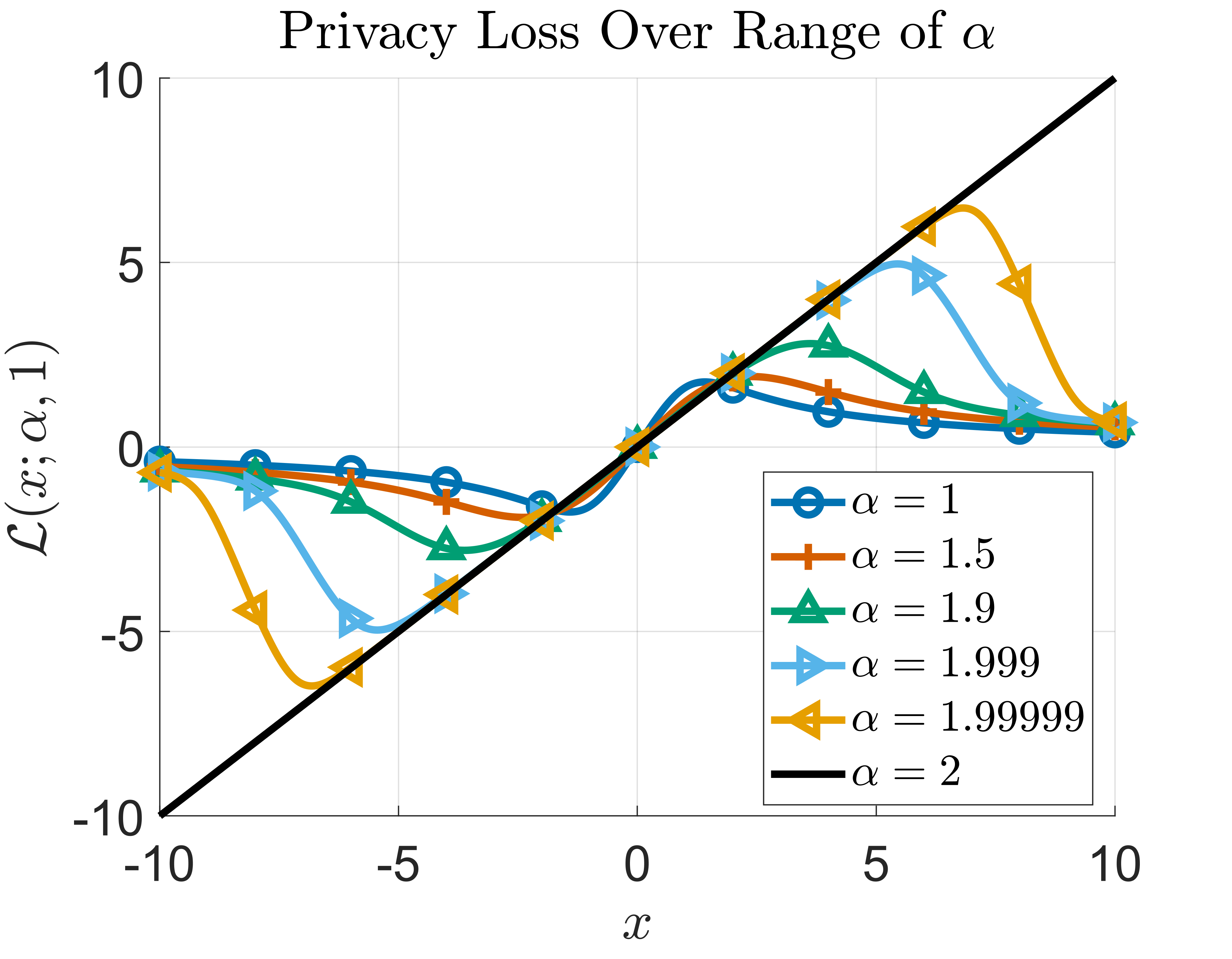}
  \caption{\label{fig:alpha_range} Denote by $\mathcal{L}(x;\alpha,\gamma=1)$ the privacy loss of the SaS mechanism with unit scale over observations $x$. Without loss of generality, let $\Dc_1$ and $\Dc_2$ be neighboring datasets such that the privacy loss of the Gaussian mechanism is linear: $\mathcal{L}(x;2,1) = x$ shown in {\bf black}. As the stability parameter $\alpha$ is reduced, we observe that the privacy loss becomes bounded, reaching a peak before converging to the $x$-axis.}
\end{figure}
In Figure \ref{fig:gamma_range}, with $\alpha$ fixed at $3/2$, we see that as the scale of the density, $\gamma$, increases, the level of privacy also increases (seen in the decreasing maximum $\varepsilon$ value; recall from \ref{eqn:epmax} that $\varepsilon = \max_x \mathcal{L}(x)$).
\begin{figure}[ht!]
  \centering
  \includegraphics[width=0.6\linewidth]{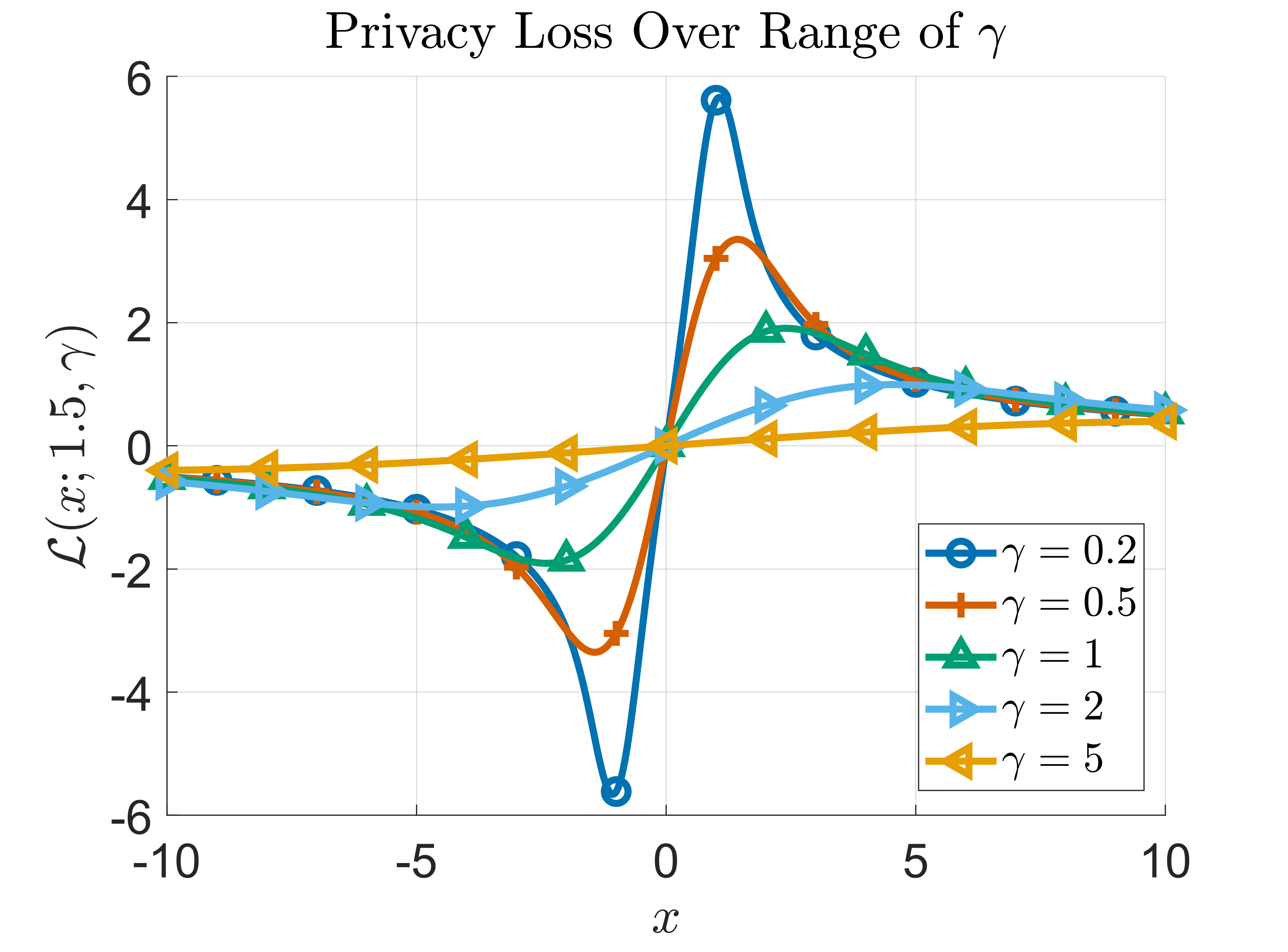}
  \caption{\label{fig:gamma_range} Denote by $\mathcal{L}(x;\alpha=1.5,\gamma)$ the privacy loss of the SaS mechanism with stability parameter $\alpha=1.5$ over observations $x$. Without loss of generality, let $\Dc_1$ and $\Dc_2$ be neighboring datasets such that the privacy loss is symmetric about the origin. As the scale $\gamma$ of the density is increased, we observe that the increase in noise decreases the maximum possible privacy loss, compressing the curve toward the $x$-axis. }
\end{figure}

% We can expand upon the visualizations in Figures \ref{fig:alpha_range} and \ref{fig:gamma_range} by incorporating the results into a contour plot as depicted in Figure \ref{fig:gamma_v_alpha}. Note that the relationship between the privacy budget $\varepsilon$ and the parameter $\alpha$ is non-monotonic for a given $\hat{\gamma}$.
% \begin{figure}[ht!]
%   \centering
%   \includegraphics[width=6in]{figures/gamma_v_alpha.png}
%   \caption{\label{fig:gamma_v_alpha} A plot showing constant $\log(\varepsilon)$ contours for a range of $\alpha$ and $\hat{\gamma}$ values. We use the $\log$ of the privacy budget to better illustrate the asymptotic trends for small $\hat{\gamma}$ values. We note there is an additional asymptote at $\alpha=2$, but that the steepness of the graph makes it difficult to capture.}
% \end{figure}
    
Next, to derive the behavior of the privacy loss at observation $x$ in terms of the scale, we use a special case of the second partial sum expansion discussed by in \cite{bergstrom52}.

\begin{lemma}
  \label{lem:fse2}
  (A Second Finite Series Expansion)
  The symmetric alpha-stable density (\ref{eqn:sas}), with $\alpha \in (1,2)$ and $\gamma=1$, has the following finite series expansion:
\begin{equation}
  \label{eqn:series2}
  \begin{aligned} 
    & p_{SaS}(x; \alpha, 1, 0) =  \\  &\frac{1}{\pi}  \sum_{k=0}^n (-1)^k \frac{\Gamma(\frac{k+1}{\alpha})}{k!\alpha} (x)^{k} \cos \bigg( \frac{k \pi}{2} \bigg) +  O\big(|x|^{n+1}\big),
  \end{aligned}
\end{equation}
as $|x| \to 0$.
\end{lemma}
\begin{proof}
  The full form provided in \cite{bergstrom52} states the result for the full range $\beta \in (-1, 1)$. In our work, we only require (\ref{eqn:series2}), so for brevity, we leave out the full form of the series.
\end{proof}

Below we make use of the following two elementary Taylor series expansions. For any $c \neq 0$:
\begin{equation}
  \label{eqn:ts1}
  \frac{1}{c + x} = \frac{1}{c} - \frac{x}{c^2} + \frac{x^2}{c^3} - \frac{x^3}{c^4} + O(x^4),
\end{equation}
and 
\begin{equation}
  \label{eqn:ts2}
  \ln \frac{c + x}{c} = \frac{x}{c} - \frac{x^2}{2c^2} + \frac{x^3}{3c^3} + \frac{x^4}{4c^4} + O(x^5).
\end{equation}

Using Lemma \ref{lem:fse2}, we now assert a relationship between the privacy loss for a given observation $x$ and the scale of the SaS mechanism $\gamma$.
\begin{theorem}
  \label{thm:gamma_prop}
  Let $\mathcal{D}_1 \simeq \mathcal{D}_2$ be neighboring datasets and $f$ a bounded query that operates on them. Denote by $\Delta_1$ the $\ell_1$-sensitivity of the query $f$. Let $\mathcal{M}_f$ be a SaS mechanism with $\alpha \in (1,2)$. Let the observation $x$ be fixed and take $\gamma$ to be the independent variable. Then
\begin{equation}
\label{eqn:loss_o}
[\mathcal{L}^{SaS}_{\mathcal{D}_1||\mathcal{D}_2}(x)](\gamma) = 
O\Big(\frac{\Delta_1}{\gamma}\Big) \textrm{ as } \gamma \to \infty.
\end{equation}
($\Delta_1$ is included in (\ref{eqn:loss_o}) in order to highlight the analogy with (\ref{eqn:other-scaling}).)
\end{theorem}
\begin{proof}
  Fix the observation $x$, then, by Lemma \ref{thm:max}, the maximum privacy loss for $x$ over the datasets $\mathcal{D}_1$ and $\mathcal{D}_2$ is
  \begin{equation}
    \label{eqn:priv-loss}
    [\mathcal{L}^{SaS}_{\mathcal{D}_1||\mathcal{D}_2}(x)](\gamma) = \ln \frac{\int\limits_{-\infty}^{\infty}e^{-|\gamma t|^\alpha - it(x-\Delta_1)}dt}{\int\limits_{-\infty}^{\infty}e^{-|\gamma t|^\alpha - itx}dt}.
  \end{equation}
  Let $\hat{t} = \gamma t$, $\hat{x} = x\Delta_1$, and  $\hat{\gamma} = \gamma/\Delta_1$ and denote $(\hat{x} - 1)/\hat{\gamma}$ and $\hat{x}/\hat{\gamma}$ by $x_1$ and $x_2$. The Eq. (\ref{eqn:priv-loss}) becomes
  \begin{equation}
    \label{eqn:fracO}
    [\mathcal{L}^{SaS}_{\mathcal{D}_1||\mathcal{D}_2}(x)](\gamma) = \ln \frac{\int\limits_{-\infty}^{\infty}e^{-|\hat{t}|^\alpha - i\hat{t}x_1}d\hat{t}}{\int\limits_{-\infty}^{\infty}e^{-|\hat{t}|^\alpha - i\hat{t}x_2}d\hat{t}}.
  \end{equation}
  We consider the numerator first, followed by the denominator. Expand the numerator in (\ref{eqn:fracO}) using the partial sum expansion given in Lemma \ref{lem:fse2} with $n=0$:
  \begin{equation}
    p_{SaS}(x_1;\alpha, 1,0) = \frac{\Gamma\big(\frac{1}{\alpha}\big)}{\alpha} + O(|x_1|), \quad |x_1| \to 0.
  \end{equation}
  For simplicity we denote
  \begin{equation}
    a(\alpha) = \frac{\Gamma\big(\frac{1}{\alpha}\big)}{\alpha},
  \end{equation}
  which gives
  \begin{equation}
    \label{eqn:num_exp}
    p_{SaS}(x_1;\alpha, 1,0) = a(\alpha) + O(|x_1|), \quad |x_1| \to 0.
  \end{equation}
  Thus, there exist positive constants $C$ and $x_0$ such that
  \begin{equation}
    \label{eqn:part-exp}
    |p_{SaS}(x_1;\alpha, 1,0)| \leq a + C|x_1|, \quad \forall |x_1| \leq x_0. 
  \end{equation}
  Replace $x_1$ by its definition in (\ref{eqn:part-exp}), first noting that the translations in $x$ are described by the last parameter in the notation for the SaS density:
  \begin{equation}
    p_{SaS}\Big(\frac{\hat{x}-1}{\hat{\gamma}};\alpha, 1, 0\Big) = p_{SaS}\Big(\frac{\hat{x}}{\hat{\gamma}};\alpha, 1,\frac{1}{\hat{\gamma}}\Big).
  \end{equation}
  Then
  \begin{equation}
    \label{eqn:another}
    \Big|p_{SaS}\Big(\frac{\hat{x}}{\hat{\gamma}};\alpha, 1,\frac{1}{\hat{\gamma}}\Big)\Big| \leq a + C\Big|\frac{\hat{x}-1}{\hat{\gamma}}\Big|, \ \ \forall \Big|\frac{\hat{x}-1}{\hat{\gamma}}\Big| \leq x_0.
  \end{equation}
  Note that the range restriction in (\ref{eqn:another}) is equivalent to $\hat{\gamma} \geq |\hat{x}-1|/x_0$. Thus, denote $|\hat{x}-1|/x_0$ and $C|\hat{x}-1|$ by $\gamma_0$ and $\hat{C}$ respectively. Then
  \begin{equation}
    \Big|p_{SaS}(\frac{\hat{x}}{\hat{\gamma}};\alpha, 1,\frac{1}{\hat{\gamma}})\Big| \leq a + \hat{C} \cdot \frac{1}{\hat{\gamma}}, \quad \forall \hat{\gamma} \geq \gamma_0,
  \end{equation}
  which can be represented in big O notation as
  \begin{equation}
    \label{eqn:bigo1}
    p_{SaS}\Big(\frac{\hat{x}}{\hat{\gamma}};\alpha, 1,\frac{1}{\hat{\gamma}}\Big) = a + O\Big(\frac{1}{\hat{\gamma}}\Big), \quad \hat{\gamma} \to \infty.
  \end{equation}
  Using the same logic, the denominator in (\ref{eqn:priv-loss}) can be represented as
  \begin{equation}
    \label{eqn:bigo2}
    p_{SaS}\Big(\frac{\hat{x}}{\hat{\gamma}};\alpha, 1,0\Big) = a + O\Big(\frac{1}{\hat{\gamma}}\Big), \quad \hat{\gamma} \to \infty.
  \end{equation}
  Combining (\ref{eqn:bigo1}) and (\ref{eqn:bigo2}), (\ref{eqn:fracO}) can now be expressed for large $\hat{\gamma}$ in the form
  \begin{equation}
    \label{eqn:loss-o}
    [\mathcal{L}^{SaS}_{\mathcal{D}_1||\mathcal{D}_2}(x)](\gamma) = \ln \frac{a + O(\frac{1}{\hat{\gamma}})}{a + O(\frac{1}{\hat{\gamma}})}, \quad  \hat{\gamma} \to \infty.
  \end{equation}
  Using the elementary Taylor series (\ref{eqn:ts1}), we rewrite the denominator as
  \begin{equation}
    \label{eqn:o-exp}
    \frac{1}{a + O(\frac{1}{\hat{\gamma}})} = \frac{1}{a} + O\Big(\frac{1}{\hat{\gamma}}\Big) = \frac{1 + O(\frac{1}{\hat{\gamma}})}{a},
  \end{equation}
  as $\hat{\gamma} \to \infty$. Substituting (\ref{eqn:o-exp}) into (\ref{eqn:loss-o}) gives
  \begin{equation}
    \begin{aligned}
      [\mathcal{L}^{SaS}_{\mathcal{D}_1||\mathcal{D}_2}(x)](\gamma) & = \ln \frac{a + O(\frac{1}{\hat{\gamma}})}{1} \cdot \frac{1 + O(\frac{1}{\hat{\gamma}})}{a} \\
                                                                    & = \ln \frac{a + O(\frac{1}{\hat{\gamma}}) + O(\frac{1}{\gamma^2})}{a},
    \end{aligned}
  \end{equation}
  as $\hat{\gamma} \to \infty$. The squared term is dominated in the limit and leaves 
  \begin{equation}
    \label{eqn:big-o2}
    [\mathcal{L}^{SaS}_{\mathcal{D}_1||\mathcal{D}_2}(x)](\gamma) = \ln \frac{a + O(\frac{1}{\hat{\gamma}})}{a}, \quad \hat{\gamma} \to \infty.
  \end{equation}
  Next, we use the elementary Taylor series (\ref{eqn:ts2}) and expand to
  \begin{equation}
    \label{eqn:big-lim}
    [\mathcal{L}^{SaS}_{\mathcal{D}_1||\mathcal{D}_2}(x)](\gamma) = \frac{O(\frac{1}{\hat{\gamma}})}{a} + O\Big(\frac{1}{\hat{\gamma}^2}\Big), \quad \hat{\gamma} \to \infty.
  \end{equation} 
  Recalling that $\hat{\gamma}= \gamma/\Delta_1$, we complete the proof:
  \begin{equation}
    [\mathcal{L}^{SaS}_{\mathcal{D}_1||\mathcal{D}_2}(x)](\gamma) =  O\Big(\frac{\Delta_1}{\gamma}\Big), \quad \gamma \to \infty.
  \end{equation}  
\end{proof}

Theorem \ref{thm:gamma_prop} only guarantees that the privacy of a specific observation $x$ scales as $O(\Delta_1/\gamma)$ for large $\gamma$. Without additional information about the location of the maximum, which is difficult to attain due to the lack of a known closed form for the general SaS density, Theorem \ref{thm:gamma_prop} does not allow us to conclude that the maximum over all observations scales in the same manor. Because of this, in Figure \ref{fig:gamma_epsilon_all} we provide numerical results graphing the max privacy loss $\varepsilon$ over a range of scale values $\gamma$ (with $\Delta_1 = 1$) for a selection of $\alpha$ values.
\begin{figure}[ht!]
  \centering
  \includegraphics[width=0.6\linewidth]{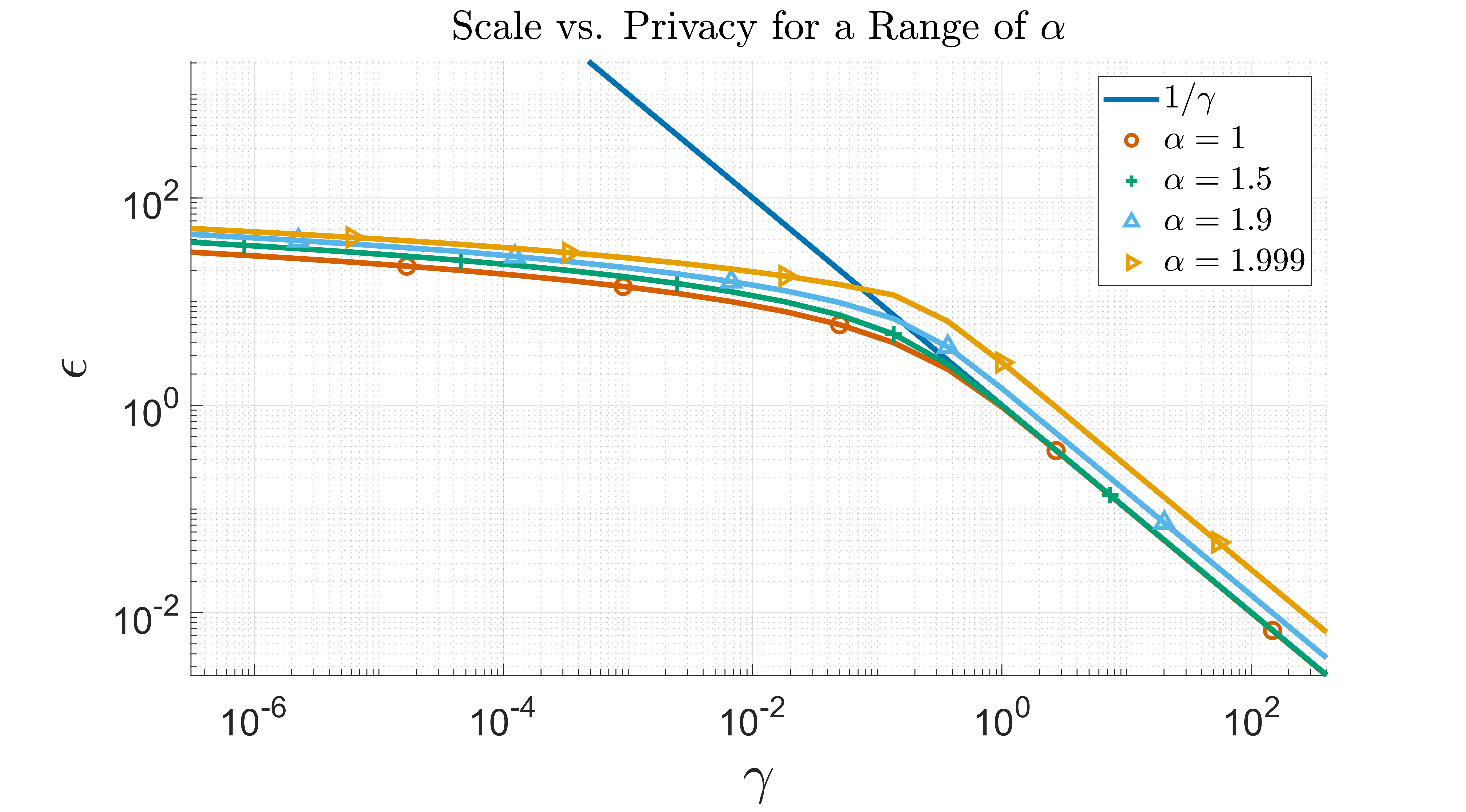}
  \caption{\label{fig:gamma_epsilon_all} The maximum privacy loss of the Laplace mechanism is inversely related to the scale of the injected noise $\gamma$, show linearly on the log-log plot in {\bf \color{blue} blue}. For large values of $\gamma$, the privacy loss of the SaS mechanism falls off at the same rate, see Corollary \ref{cor:bigg}. However, for small scales, as $\gamma$ is decreased, the SaS mechanism increases at a rate of $O(\log(1/\gamma))$ as shown in Corollary \ref{cor:smol} (as opposed to $O(1/\gamma)$ for the Laplace). The equation for the Cauchy's privacy loss, shown here in {\bf \color{Bittersweet} orange}, is explicitly given in Equation (\ref{eqn:cauchy-loss}) with $\Delta_1=1$.}
\end{figure}
Note that because this is a log-log plot, a vertical shift, as seen with $\alpha=1.9$, corresponds to a multiplicative scalar in the limiting behavior. We observe that as $\gamma$ increases, the maximum privacy loss falls off at a rate similar to that for $\alpha=1$ (at least for practically useful values of $\varepsilon$ and $\gamma$). To this end, we take advantage of the closed form of the density when $\alpha=1$ to provide more concrete results for that case.

\begin{theorem}
  Let $\mathcal{D}_1 \simeq \mathcal{D}_2$ be neighboring datasets and $f$ a bounded query  with $\ell_1$-sensitivity $\Delta_1$ that operates on them. Take the stability parameter $\alpha=1$ for the SaS mechanism. Then, the privacy budget $\varepsilon$ as a function of the scale $\gamma$ is given by
  \begin{equation}
    \label{eqn:cauchy-loss}
    \varepsilon(\gamma) = \ln \frac{\sqrt{4(\frac{\gamma}{\Delta_1})^2+1}+1}{\sqrt{4(\frac{\gamma}{\Delta_1})^2+1}-1}.
  \end{equation}
\end{theorem}
\begin{proof}
  Consider the privacy budget $\varepsilon$ for the SaS mechanism when $\alpha = 1$:
  \begin{equation}
    \varepsilon := \max_{x\in\mathbb{R}} \mathcal{L}^{SaS}_{\mathcal{D}_1||\mathcal{D}_2}(x).
  \end{equation} 
  As in the proof of the foregoing theorem, let $\hat{t} = \gamma t$, $\hat{x} = x\Delta_1$, and  $\hat{\gamma} = \gamma/\Delta_1$ in (\ref{eqn:priv-loss}):
  \begin{equation}
    \label{eqn:frac1b}
    [\mathcal{L}^{SaS}_{\mathcal{D}_1||\mathcal{D}_2}(x)](\gamma) = \ln \frac{\int\limits_{-\infty}^{\infty}e^{-|\hat{t}|^\alpha - i\hat{t}\frac{\hat{x}-1}{\hat{\gamma}}}d\hat{t}}{\int\limits_{-\infty}^{\infty}e^{-|\hat{t}|^\alpha - i\hat{t}\frac{\hat{x}}{\hat{\gamma}}}d\hat{t}}.
  \end{equation}
  Note that the SaS density, when $\alpha=1$, takes the closed form 
  \begin{equation}
    \label{eqn:cauchy}
    p_{SaS}(x;1,\gamma, \mu) = \frac{1}{\pi\gamma(1+(\frac{x - \mu}{\gamma})^2)}.
  \end{equation}
  Substituting (\ref{eqn:cauchy}) into (\ref{eqn:frac1b}) gives
  \begin{equation}
    \label{eqn:cauchy-full}
    [\mathcal{L}^{SaS}_{\mathcal{D}_1||\mathcal{D}_2}(x)](\gamma)  = \ln\frac{1 + (\frac{\hat{x}}{\hat{\gamma}})^2}{1 + (\frac{\hat{x}-1}{\hat{\gamma}})^2}.
  \end{equation}
  To find the maximum, we take the derivative of the right side with respect to $\hat{x}$,
  \begin{equation}
    \frac{d}{d\hat{x}} \ln\frac{1 + (\frac{\hat{x}}{\hat{\gamma}})^2}{1 + (\frac{\hat{x}-1}{\hat{\gamma}})^2} = \frac{-2(\hat{x}^2-\hat{x} - \hat{\gamma}^2)}{(\hat{\gamma}^2 + (\hat{x}-1)^2)(\hat{\gamma}^2 + \hat{x}^2)}.
  \end{equation}
  This equates to $0$ when
  \begin{equation}
    \hat{x}^2 - \hat{x} - \hat{\gamma}^2 = 0.
  \end{equation}
  There are two solutions:
  \begin{equation}
    \hat{x}^* = \frac{1}{2}(1 \pm \sqrt{1 + 4\hat{\gamma}^2}).
  \end{equation}
  Since the privacy loss is symmetric, we take the positive solution without loss of generality. Substituting the positive maximum location into (\ref{eqn:cauchy-full}) gives
  \begin{equation}
    \label{eqn:almost}
    \varepsilon(\hat{\gamma}) = \ln\frac{1 + \frac{1 + \sqrt{1 + 4\hat{\gamma}^2}}{4\hat{\gamma}^2}}{1 + \frac{1+4\hat{\gamma}^2}{4\hat{\gamma}^2}}.
  \end{equation}
  Recalling that $\hat{\gamma} = \gamma/\Delta_1$, equation (\ref{eqn:almost}) is equivalent to the following expression after simplification,
  \begin{equation}
    \label{eqn:alpha_one}
    \varepsilon\big|_{\alpha=1}(\gamma) = \ln \frac{\sqrt{4(\frac{\gamma}{\Delta_1})^2+1}+1}{\sqrt{4(\frac{\gamma}{\Delta_1})^2+1}-1}.
  \end{equation}   
\end{proof}

To study the limiting behavior of the privacy loss, we again invoke three elementary Taylor series:
\begin{equation}
  \label{eqn:ts3}
  \sqrt{1+x^2} \pm x = 1 \pm x + \frac{x^2}{2} - \frac{x^4}{8} + O(x^6),
\end{equation}
\begin{equation}
  \label{eqn:ts4}
  \ln\frac{(1+x)}{(1-x)} = 1 + 2x + \frac{2x^3}{3} + \frac{2x^5}{5} + O(x^7),
\end{equation}
and 
\begin{equation}
  \label{eqn:ts5}
  \sqrt{4x^2+1}+c = (c+1) + 2x^2 -2x^4 + O(x^5).
\end{equation}

\begin{corollary} (Large scale approximation)
  \label{cor:bigg}
  In the limit as $\gamma$ grows without bound, when $\alpha=1$ the privacy budget $\varepsilon$, falls off at the following rate:
  \begin{equation}
    \varepsilon(\gamma)\big|_{\alpha=1} \approx \frac{\Delta_1}{\gamma}, \ \text{ as } \gamma \to \infty.
  \end{equation}
\end{corollary}
\begin{proof}
  The change of variables $x = \Delta_1/(2\gamma)$ applied to Eq. (\ref{eqn:alpha_one}) gives
  \begin{equation}
    \label{eqn:why}
    \varepsilon(x)\big|_{\alpha=1} = \ln \frac{\sqrt{\frac{1}{x^2} + 1}+1}{\sqrt{\frac{1}{x^2} + 1}-1}.
  \end{equation}
  Because $x$ only equates to $0$ in the limit of $\gamma\to\infty$, and we seek the dynamics when $\gamma$ is large but finite, we can safely multiply the argument of the logarithm in (\ref{eqn:why}) by $x/x$ giving
  \begin{equation}
    \label{eqn:why2}
    \varepsilon(x)\big|_{\alpha=1} = \ln \frac{\sqrt{1+x^2}+x}{\sqrt{1+x^2}-x}.
  \end{equation}
  Expand the numerator and denominator of (\ref{eqn:why2}) using the elementary Taylor series (\ref{eqn:ts3}), giving the following expression for small $x$ after eliminating the higher order terms:
  \begin{equation}
    % \label{eqn:why3}
    \varepsilon(x)\big|_{\alpha=1} \approx \ln \frac{1+x}{1-x}, \ \  \textrm{ as } x \to 0.
  \end{equation}
  This can be further simplified by appealing to the Taylor series expansion (\ref{eqn:ts4}) yielding a first order approximation
  \begin{equation}
    % \label{eqn:why3}
    \varepsilon(x)\big|_{\alpha=1} \approx 2x, \ \ \textrm{ as } x \to 0.
  \end{equation}
  Recalling that $x = \Delta_1/(2\gamma)$ now gives
  \begin{equation}
    \label{eqn:why3}
    \varepsilon(\gamma)\big|_{\alpha=1} \approx \frac{\Delta_1}{\gamma}, \ \text{ as } \gamma \to \infty.
  \end{equation}
\end{proof}

\begin{corollary} (Small scale approximation)
  \label{cor:smol}
  In the limit as $\gamma$ becomes vanishing small, for $\alpha=1$ the privacy budget $\varepsilon$, increases at the following rate:
  \begin{equation}
    \varepsilon(\gamma)\big|_{\alpha=1} \approx 2 \ln \frac{\sqrt{2}\Delta_1}{\gamma}, \text{ as } \gamma \to 0.
  \end{equation}
\end{corollary}
\begin{proof}
  Begin by expanding the argument of the logarithm in (\ref{eqn:alpha_one}) using the elementary Taylor series (\ref{eqn:ts4}),
  \begin{equation}
    \ln \frac{\sqrt{4(\frac{\gamma}{\Delta_1})^2+1}+1}{\sqrt{4(\frac{\gamma}{\Delta_1})^2+1}-1} = \ln \frac{2+2(\frac{\gamma}{\Delta_1})^2 + O(\gamma^3)}{2(\frac{\gamma}{\Delta_1})^2 + O(\gamma^3)}, \quad \gamma \to 0.
  \end{equation}
  As $\gamma$ tends to $0$, the higher order behavior is dominated by $\gamma^2$ and we have
  \begin{equation}
    \label{eqn:somthin}
    \varepsilon(\gamma)\big|_{\alpha=1}  \approx \ln \frac{2}{(\frac{\gamma}{\Delta_1})^2}, \quad \gamma \to 0.
  \end{equation}
  Equivalently, the expression in (\ref{eqn:somthin}) gives
  \begin{equation}
    \label{eqn:somthin2}
    \varepsilon(\gamma)\big|_{\alpha=1}  \approx 2 \ln \frac{\sqrt{2}\Delta_1}{\gamma}, \text{ as } \gamma \to 0.
  \end{equation}
  \vspace{-2mm}
\end{proof}

In Figure \ref{fig:gamma_epsilon_one}, we supplement the numerical results with graphs of the limiting behavior derived in Corollaries \ref{cor:bigg} and \ref{cor:smol}.
\begin{figure}[ht!]
  \centering
  \includegraphics[width=0.6\linewidth]{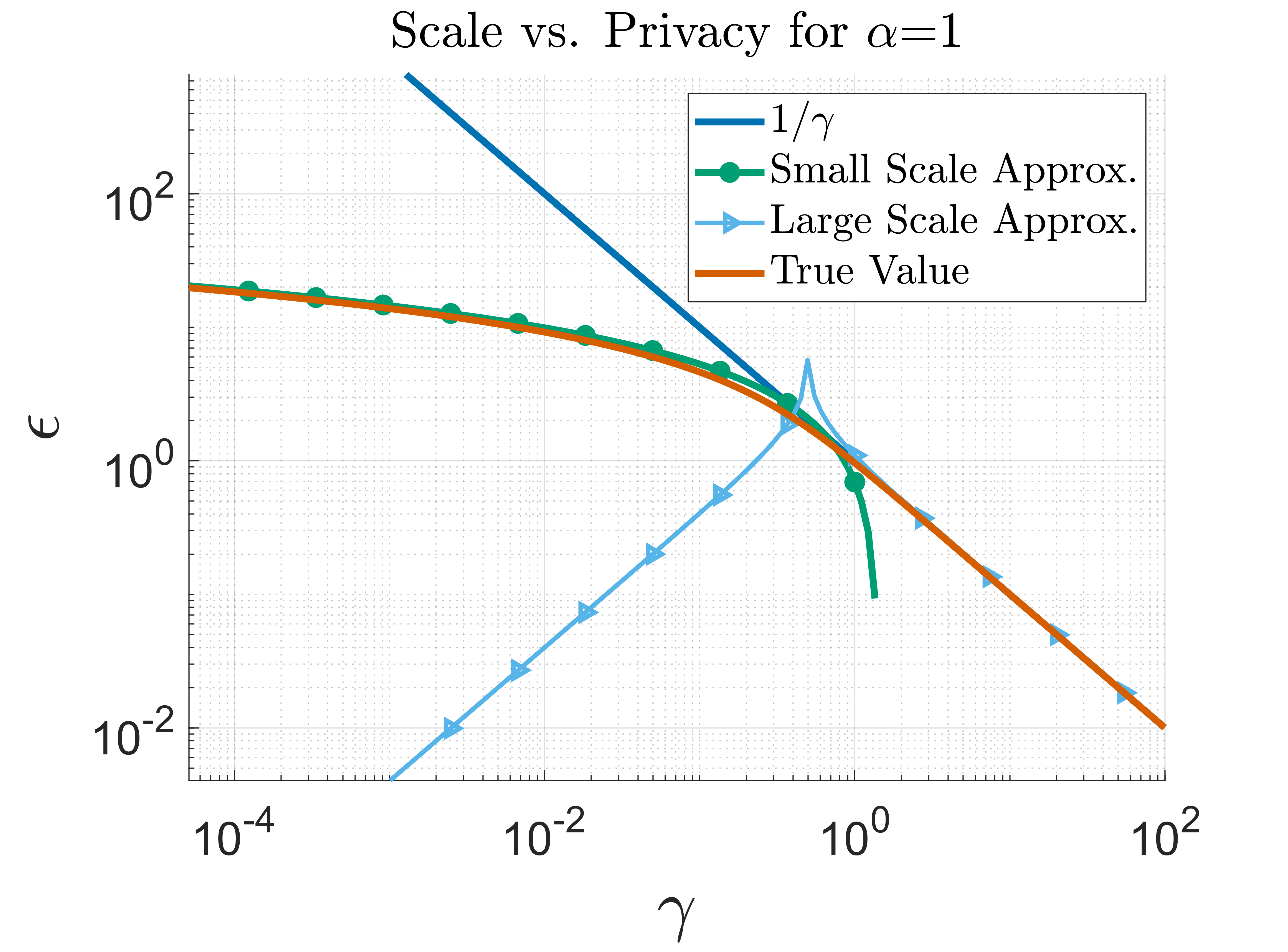}
  \caption{\label{fig:gamma_epsilon_one} The privacy loss $\varepsilon$ for the SaS mechanism, with $\alpha = 1$ and $\Delta_1=1$, over a range of scale values $\gamma$ described by (\ref{eqn:cauchy-loss}) and shown here in {\bf \color{Bittersweet} orange}. For small $\gamma$, the privacy loss is approximated by Eq. (\ref{eqn:somthin2}) in {\bf \color{ForestGreen} green}, and, for large $\gamma$, the privacy loss is approximated by Eq. (\ref{eqn:why3}), shown in {\bf \color{SkyBlue} light blue}. For comparison, the privacy loss of the Laplace mechanism is shown in {\bf \color{blue} blue}.}
\end{figure}
The figure numerically confirms that, for small $\gamma$, the SaS mechanism, due the appearance of the logarithm in (\ref{eqn:somthin2}), scales better than the Laplace and Gaussian mechanisms as recalled in (\ref{eqn:other-scaling}).

Now that we have shown that the SaS mechanism behaves in a similar manner to other common privacy mechanisms, we move on to describe the expected error that the SaS mechanisms introduces into the query's result by using any of these mechanisms.

\section{Error Analysis}
\label{sec:error}
It is typical for methods to employ the $\ell_2$-norm when defining a measure of error. However, the moments of SaS densities are only defined up to $\alpha$, and since we consider $\alpha < 2$, the second moment lacks a clear definition \cite{nolan20}. In lieu of the $\ell_2$-norm, we opt for the mean absolute deviation (MAD), as used in \cite{roth14}:

\begin{definition}
  \label{eqn:exp-dist}
  (Expected Privacy Distortion) Let $\mathcal{D}$ be a dataset and denote by $f(\mathcal{D})$ and $\mathcal{M}_f(\mathcal{D})$ the response of a query and privacy mechanism respectively. Denote the density of the privacy mechanism by $Y$. The mean absolute deviation is
  \begin{equation}
    E\big(f(\mathcal{D}), \mathcal{M}_f(\mathcal{D})\big) := \mathbb{E} |f(\mathcal{D}) - \mathcal{M}_f(\mathcal{D})|,
  \end{equation}
  which is equivalent to the expectation of the absolute value of the injected noise $Y$:
  \begin{equation}
    E\big(f(\mathcal{D}), \mathcal{M}_f(\mathcal{D})\big) = \mathbb{E} |Y|.
  \end{equation}
  \vspace{-6mm}    
  \begin{equation*}\tag*{\textrm{$\blacktriangleleft$}}\end{equation*}
\end{definition}

Before beginning the analysis of the error incurred by the SaS mechanism, we establish that the SaS mechanism adheres to strict stability.
\begin{lemma}
  \label{lem:strict}
  (SaS density is \textit{Strictly} Stable) The SaS density (\ref{eqn:sas}) with location parameter $\mu=0$ is strictly stable.
\end{lemma}
\begin{proof}
Consider three independent and identically distributed SaS densities denoted by $Y_1$, $Y_2$, and $Y$ with $\mu=0$. Let $a$ and $b$ represent two scalar values. Next, examine the density of the combined random variable $aY_1 + bY_2$. Since SaS densities are determined by their characteristic functions, we establish the following relation:
\begin{equation}
  \varphi_{aY_1+bY_2}(t) = \varphi_{aY_1}(t) \varphi_{bY_2}(t).
\end{equation}
Using the definition of a characteristic function, we bring the constants into the argument
\begin{equation}
  \begin{aligned}
    \varphi_{aY_1}(t) \varphi_{bY_2}(t) & = \mathbb{E}[e^{itaY_1}]\mathbb{E}[e^{itbY_2}] \\ 
                                          & = \varphi_{Y_1}(at) \varphi_{Y_2}(bt).
  \end{aligned}
\end{equation}
Expand by substituting the expression for the characteristic function of a stable distribution with $\mu=0$  (\ref{eqn:char}) into both functions on the right side,
\begin{equation}
  \begin{aligned}
    \varphi_{Y_1}(at) \varphi_{Y_2}(bt) & = \exp(|\gamma a t|^\alpha)\exp(|\gamma bt|^\alpha)\\
                                          & =  \exp{ |(a^\alpha+b^\alpha)^{1/\alpha}\gamma t|^\alpha}.
  \end{aligned}
\end{equation}
Setting $c = (a^\alpha + b^\alpha)^{1/\alpha}$ gives $aY_1 + bY_2 = cY$.
\end{proof}
We are now equipped to determine the expected error introduced into the query's response by the SaS mechanism.
\begin{theorem}
\label{thm:distortion}
    (Expected Distortion Due to SaS mechanism) Let $f$ be a bounded query that operates on dataset $\mathcal{D}$. Denote by $\mathcal{M}_f$ the SaS mechanism and take the stability parameter $\alpha$ to be restricted to the range $\alpha\in(1,2)$. Then, the mean absolute distortion is
    \begin{equation}
        E\big(f(\mathcal{D}, \mathcal{M}_f(\mathcal{D})\big) = \frac{2\gamma}{\pi}\Gamma\big(1-\frac{1}{\alpha}\big).
    \end{equation}
\end{theorem}
\begin{proof}
    Note that, by Lemma \ref{lem:strict}, the noise injected by the SaS mechanism is strictly stable. In \cite{nolan20}, the proof of Corollary 3.5 includes a statement that if a density $Y$ is strictly stable, then its mean absolute deviation is given by
    \begin{equation}
    \label{eqn:mad-sas}
        \mathbb{E}[|Y|] = \frac{2\gamma}{\pi}\Gamma\big(1-\frac{1}{\alpha}\big).
    \end{equation}
\end{proof}

We now provide the expected distortions of the two most common privacy mechanisms from \cite{roth14, dwork06b}: the Laplace and the Gaussian mechanisms , to show that each induces an error linear with the scale of the noise. The mean absolute deviation of the Laplace density is
\begin{equation}
    \mathbb{E}[|Lap(0, b)|] = \mathbb{E} [Exp(b^{-1})] = b.
\end{equation}
The mean absolute deviation Gaussian density is the expected value of the half-normal random variable 
\begin{equation}
    \mathbb{E}[|\mathcal{N}(0, \sigma^2)|] = \sqrt{\frac{2}{\pi}}\sigma.
\end{equation}
Note that for each of the three densities, the error is related linearly to the density's respective scale. From (\ref{eqn:mad-sas}) we recover the distortion of the Gaussian mechanism by taking $\alpha=2$ and $\gamma= \sigma/\sqrt{2}$. Next, we proceed to prove that the expected distortion is monotonic in $\alpha$, reaching a minimum when $\alpha=2$ and diverging as $\alpha$ tends to $1$.

\begin{remark}
  We note that their is a linear relationship between the expected error of the SaS mechanism $\mathbb{E}[|Y^{SaS}|]$ and the scale of the density $\gamma$ in (\ref{eqn:mad-sas}). Recall that the privacy budget $\varepsilon$ is inversely proportional to large values of $\gamma$ (Figure \ref{fig:gamma_epsilon_all}). This relationship is proven for $\alpha=1$ in Corollary \ref{cor:bigg}. Therefore, small values of $\varepsilon$, while enhancing the client's privacy, necessarily increase the expected error induced by the mechanism. In other words, the level of privacy is inversely related to the accuracy of the query. We note that this relationship is shared by the other common mechanisms.
\end{remark}

\begin{corollary}
\label{cor:mon-error}
    (Error is Monotonic in $\alpha$) The mean absolute distortion injected into a query by the SaS mechanism decreases monotonically as $\alpha$ increases from $1$ to $2$.
\end{corollary}
\begin{proof}
Because the stability parameter $\alpha$ is chosen from the bounded set $(1,2)$, the argument of the Gamma function in (\ref{eqn:mad-sas}) varies between $(0,1/2)$. The Gamma function has an asymptote at $x=0$ and reaches a local minimum in the right plane at $x \approx 1.462$, see \cite{oeis}. Thus, for a given $\gamma$, the distortion in Eq. (\ref{eqn:mad-sas}) is minimized when $\alpha$ tends to $2$.
\end{proof}

This minimum is proven in Corollary \ref{cor:mon-error} and depicted in Figure \ref{fig:gamma}.
\begin{figure}[ht!]
  \centering
  \includegraphics[width=0.6\linewidth]{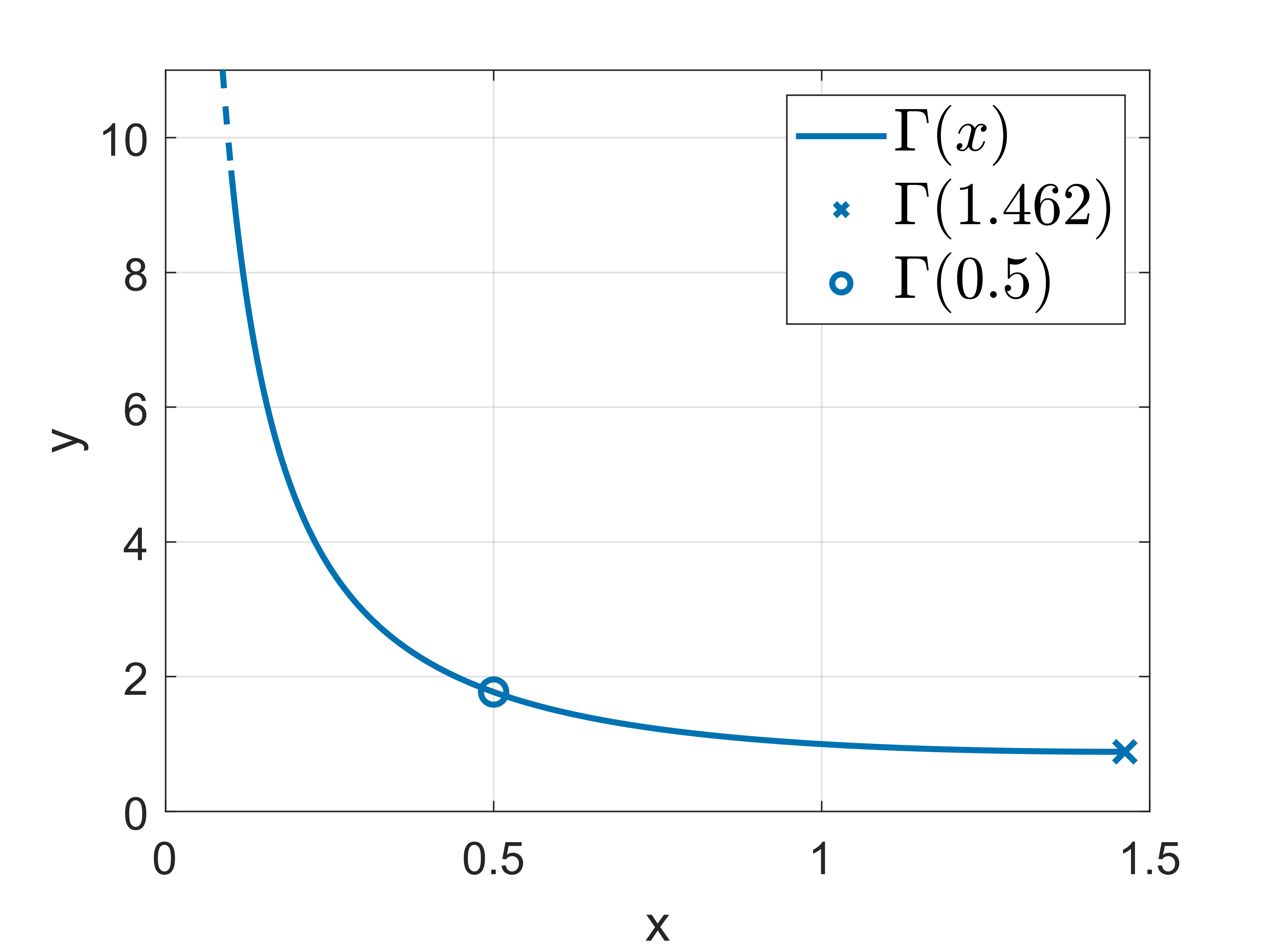}
  \caption{\label{fig:gamma} The Gamma function $\Gamma(x)$ achieves a minimum value in the right hand plane at $x\approx1.462$. When $\alpha$ is bounded between $[1,2)$, the Gamma component of the injected error takes input values in $[0,0.5)$. The Gamma function is monotonically decreasing in this interval from $\infty$ to $1.7725$.}
\end{figure}
A naive first thought is thus that $\alpha=2$ is the \textit{optimal} value for the parameter $\alpha$ as the injected error achieves a minimum. However, this is not necessarily the case as choosing $\alpha=2$ increases the required scale $\gamma$ necessary to achieve a given privacy budget $\varepsilon$. Table \ref{tab:error1} provides a list of expected distortions for a selection of stability values $\alpha$. By selecting $\alpha$ close to $2$ we can achieve an essentially equivalent expected distortion but provide better levels of privacy.

\begin{table}[ht!] 
  \caption{\label{tab:error1} The expected distortion for the Gaussian mechanism ($\alpha=2$) and a selection of SaS mechanisms ($\alpha < 2$). The expected distortion is given as a multiple of the injected noise scale $\gamma$.}
  \vspace{1mm}
  \begin{tabular}{ll}
    \hline \rule{0pt}{2ex} 
$\alpha$                   & Expected Distortion         \\ \hline 
\multicolumn{1}{l|}{\rule{0pt}{2ex} 2}     & $1.1284\gamma$              \\
\multicolumn{1}{l|}{1.999} & $1.1289\gamma$              \\
\multicolumn{1}{l|}{1.99}  & $1.1340\gamma$              \\
\multicolumn{1}{l|}{1.95}  & $1.1576\gamma$              \\
\multicolumn{1}{l|}{1.9}   & $1.1903\gamma$              \\
\multicolumn{1}{l|}{1.8}   & $1.2687\gamma$              \\
\multicolumn{1}{l|}{1.0}   & $\infty$                    \\ \hline
\end{tabular}
\centering
\end{table}
In particular, note that the expected distortion between $\alpha=2$ and $\alpha=1.999$ differ only by only $0.044\%$. Thus, to achieve similar accuracy results to the Gaussian mechanism, we can choose to focus on $\alpha$ close to $2$, leaving a further exploration of an optimal choice of $\alpha$ for future work. It should be noted that decreasing $\alpha$ enhances the privacy experienced by the client, thus decreasing $\varepsilon$. Therefore, for a fixed privacy budget $\varepsilon$, one can compose an optimization that minimizes a weighted sum of privacy and injected noise. This optimization is problem specific since it deals with the particular sensitivity of the dataset in question, as well as the desired weighting between error and privacy. We leave it for a more application focused examination.

\section{Concluding Remarks}
\label{sec:conclusion}
We have presented, the SaS mechanism represents and shown how it advances the field of Differential Privacy. This mechanism not only provides strong guarantees of privacy but also offers distinct advantages when compared to other common privacy mechanisms. We proved that the SaS mechanism achieves pure Differential Privacy, ensuring that individual data points remain protected even in the face of powerful adversaries. This is starkly contrasted with the Gaussian mechanism, which only achieves approximate Differential Privacy. Additionally, the SaS mechanism utilizes a stable density, allowing it to be used in local applications where the Laplace mechanisms is difficult to analyze.

We showed that the expected distortion introduced by the SaS mechanism into query results can be made essentially equivalent to that of the Gaussian mechanism. The expected distortion can additionally be formulated as a compromise between injected error and privacy guarantees. As a result, we conclude that there is little reason to use the Gaussian mechanism over the SaS mechanism.
\section*{ACKNOWLEDGMENT}
This work was supported in part by Grant ECCS-2024493 from the U.S. National Science Foundation.

% \bibliographystyle{unsrt}
% \bibliography{biblio}{}

\vskip 0.2in
\bibliography{sample}

\end{document}